\documentclass[11pt]{article}
\usepackage{amsmath}
\usepackage{amsthm}
\usepackage{fixltx2e}
\usepackage{fullpage}

\usepackage{listings}
\usepackage{graphicx}
\usepackage{multirow}
\usepackage{amssymb}
\usepackage{verbatim}
\usepackage{enumitem}
\usepackage{color}

\usepackage{comment}

\newtheorem{theorem}{Theorem}[section]
\newtheorem{definition}{Definition}
\newtheorem{lemma}[theorem]{Lemma}
\newtheorem{corollary}[theorem]{Corollary}
\newtheorem{proposition}[theorem]{Proposition}
\newtheorem{claim}[theorem]{Claim}

\def\nat{{\mathbb N}}
 \def\real{{\mathbb R}}

\newcommand{\Until}{{\: \mbox{\large $\mathrel{\mathsf{U}}$} }}

\begin{document}
\normalsize

\title{Reachability for Branching Concurrent Stochastic Games}
\author{Kousha Etessami\thanks{\tt kousha@inf.ed.ac.uk}
\\U. of Edinburgh
\and 
Emanuel Martinov\thanks{\tt eo.martinov@gmail.com}\\U. of Edinburgh 
\and
Alistair Stewart\thanks{\tt stewart.al@gmail.com}\\ USC
\and
Mihalis Yannakakis\thanks{\tt mihalis@cs.columbia.edu}\\Columbia U.
}

\date{}
\maketitle

\begin{abstract}
  We give polynomial time algorithms for deciding almost-sure and
  limit-sure reachability in Branching Concurrent Stochastic Games
  (BCSGs). These are a class of infinite-state imperfect-information
  stochastic games that generalize both finite-state concurrent
  stochastic reachability games (\cite{deAHK07}),
  as well as branching simple stochastic reachability games
  (\cite{ESY-icalp15-IC}). 
\end{abstract}

\section{Introduction}

{\em Branching Processes} (BP) are infinite-state stochastic processes
that model the stochastic evolution of a population of entities of
distinct types. In each generation, every entity of each type $t$
produces a set of entities of various types in the next generation
according to a given probability distribution on offsprings for the
type $t$. BPs are fundamental stochastic models that have been used to
model phenomena in many fields, including biology (see,
e.g., \cite{KA02}), population genetics (\cite{HJV05}), physics and
chemistry (e.g., particle systems, chemical chain reactions), medicine
(e.g. cancer growth \cite{Bozic13,RBCN13}), marketing, and others.  In
many cases, the process is not purely stochastic but there is the
possibility of taking actions (for example, adjusting the conditions
of reactions, applying drug treatments in medicine, advertising in
marketing, etc.) which can influence the probabilistic evolution of
the process to bias it towards achieving desirable objectives.  Some
of the factors that affect the reproduction may be controllable (to
some extent) while others are not and also may not be sufficiently
well-understood to be modeled accurately by specific probability
distributions, and thus it may be more appropriate to consider their
effect in an adversarial (worst-case) sense.  {\em Branching
  Concurrent Stochastic Games} (BCSG) are a natural model to represent
such settings.  There are two players, who have a set of available
actions for each type $t$ that affect the reproduction for this type;
for each entity of type $t$ in the evolution of the process, the two
players select concurrently an action from their available set
(possibly in a randomized manner) and their choice of actions
determines the probability distribution for the offspring of the
entity.  The first player represents the controller that can control
some of the parameters of the reproduction and the second player
represents other parameters that are not controlled and are treated
adversarially.  The first player wants to select a strategy that
optimizes some objective.  In this paper we focus on {\em reachability
  objectives}, a basic and natural class of objectives. Some types are
designated as undesirable (for example, malignant cells), in which
case we want to minimize the probability of ever reaching any entity
of such type.  Or conversely, some types may be designated as
desirable, in which case we want to maximize the probability of
reaching an entity of such a type.

BCSGs generalize the purely stochastic Branching Processes as well as Branching Markov Decision Processes
(BMDP) and Branching Simple Stochastic Games (BSSG)
which were studied for reachability objectives in \cite{ESY-icalp15-IC}.
In BMDPs there is only one player who aims to maximize or minimize a reachability objective.
In BSSGs there are two opposing players but they control different types.  These models were studied previously also under another
basic objective, namely the optimization of {\em extinction probability},
i.e., the probability that the process will eventually become extinct, 
that is, that the population will become empty \cite{esy-icalp12, rcsg2008}.
We will later discuss in detail the prior results in these models 
and compare them with the results in this paper.

BCSGs can also be seen as a
generalization of finite-state concurrent games \cite{deAHK07}
(see also \cite{Everett57}), namely
the extension of such finite games with branching.
Concurrent games have been used in the verification area to model the
dynamics of open systems, where one player represents the system and
the other player the environment. Such a system moves sequentially
from state to state depending on the actions of the two players (the
system and the environment).  Branching concurrent games model the
more general setting in which processes can spawn new processes that
proceed then independently in parallel (e.g.. new threads are created
and terminated).  We note incidentally that even if there are no
probabilities in the system itself, in the case of concurrent games,
probabilities arise naturally from the fact that the optimal
strategies are in general randomized; as a consequence it can be shown
that branching concurrent stochastic games are expressively and
computationally equivalent to the non-stochastic version (see
\cite{rcsg2008}).

We now summarize our main results and compare and contrast them with
previous results on related models. First, we show that a Branching
concurrent stochastic game $G$ with a reachability objective has a
well-defined value, i.e., given an initial (finite) population $\mu$
of entities of various types and a target type $t^*$, if the sets of
(mixed) strategies of the two players are respectively $\Psi_1$,
$\Psi_2$, and if $\Upsilon_{\sigma, \tau}(\mu,t^*)$ denotes the
probability of reaching eventually an entity of type $t^*$ when
starting from population $\mu$ under strategy $\sigma \in \Psi_1$ for
player 1 and strategy $\tau \in \Psi_2$ for player 2, then
$\inf_{\sigma \in \Psi_1} \sup_{\tau \in \Psi_2} \Upsilon_{\sigma,
  \tau}(\mu,t^*)$ $= \sup_{\tau \in \Psi_2} \inf_{\sigma \in \Psi_1}
\Upsilon_{\sigma, \tau}(\mu,t^*)$, which is the value $v^*$ of the
game.  Furthermore, we show that the player who wants to minimize the
reachability probability always has an optimal (mixed) {\em static
  strategy} that achieves the value, i.e., a strategy $\sigma^*$ which
uses for all entities of each type $t$ generated over the whole
history of the game the same probability distribution on the available
actions, independent of the past history, and which has the property
that $v^* = \sup_{\tau \in \Psi_2} \Upsilon_{\sigma^*,
  \tau}(\mu,t^*)$.  The optimal strategy in general has to be mixed
(randomized); this was known to be the case even for finite-state
concurrent games \cite{deAHK07}.  On the other hand, the player that
wants to maximize the reachability probability of a BCSG may not have
an optimal strategy (whether static or not), and it was known that
this holds even for BMDPs, i.e., even when there is only one player
\cite{ESY-icalp15-IC}.  This also holds for finite-state CSGs: the
player aiming to maximize reachability probability does not
necessarily have any optimal strategy \cite{deAHK07}.

To analyze BCSGs with respect to reachability objectives, we model
them by a system of equations $x=P(x)$, called a {\em minimax
  Probabilistic Polynomial System} (minimax-PPS for short), where $x$
is a tuple of variables corresponding to the types of the BCSG.  There
is one equation $x_i=P_i(x)$ for each type $t_i$, where $P_i(x)$ is
the value of a (one-shot) two-player zero-sum matrix game, whose payoff for
every pair of actions is given by a polynomial in $x$ whose
coefficients are positive and sum to at most 1 (a probabilistic
polynomial). The function $P(x)$ defines a monotone operator from
$[0,1]^n$ to itself, and thus it has, in particular, 
a {\em greatest fixed point}
(GFP) $g^*$ in $[0,1]^n$. We show that the coordinates $g^*_i$ of the
GFP give the optimal {\em non-reachability} probabilities for the BCSG
game when started with a population that consists of a single entity
of type $t_i$. The value of the game for any initial population $\mu$
can be derived easily from the GFP $g^*$ of the minimax-PPS.
This generalizes a result in \cite{ESY-icalp15-IC},
which established an analogous result for the special case of BSSGs.
It also follows from our minimax-PPS equational characterization
that {\em quantitative} decision problems for BCSGs, such as deciding
whether the reachability game value is $\geq p$ for a given 
$p \in (0,1)$ are all solvable in PSPACE.

Our main algorithmic results concern the qualitative analysis of the
reachability problem, that is, the problem of determining whether one of
the players can win the game with probability 1, i.e.,  if the value of
the game is 0 or 1.  We provide the first polynomial-time algorithms
for qualitative reachability analysis for branching concurrent
stochastic games.  For the value=0 problem, the algorithm and its
analysis are rather simple.  If the value is 0, the algorithm computes
an optimal strategy $\sigma^*$ for the player that wants to minimize
the reachability probability; the constructed strategy $\sigma^*$ is
in fact static and deterministic, i.e., it selects for each type
deterministically a single available action, and guarantees
$\Upsilon_{\sigma^*, \tau}(\mu,t^*)=0$ for all $\tau \in \Psi_2$.  If the
value is positive then the algorithm computes a static mixed strategy
$\tau$ for the player maximizing reachability probability 
that guarantees $\inf_{\sigma \in \Psi_1} \Upsilon_{\sigma, \tau}(\mu,t^*) >0$.

The value=1 problem is much more complicated.
There are two versions of the value=1 problem,
because it is possible that the value of the game is 1 but there
is no strategy for the maximizing player that guarantees reachability with probability 1. The critical reason for this is the
concurrency in the moves of the two players: for BMDPs and BSSGs, it is known that if the value is 1 then there is a strategy $\tau$ that achieves it \cite{ESY-icalp15-IC};\footnote{When the value is positive and {\em not} equal
to 1, even for BMDPs
there need not exist an optimal strategy for the player maximizing
reachability probability \cite{ESY-icalp15-IC}.}
 on the other hand, this is not 
the case even for finite-state concurrent games \cite{deAHK07}.
Thus, we have two versions of the problem.
In the first version, called the
{\em almost-sure problem}, we want to determine whether there exists a strategy $\tau^*$ for player 2 that guarantees that the
target type $t^*$ is reached with probability 1 regardless of the
strategy of player 1, i.e.,  such that $\Upsilon_{\sigma, \tau^*}(\mu,t^*) =1$
for all $\sigma \in \Psi_1$.
In the second version of the problem, called the {\em limit-sure problem}, we want to determine if the value 
$v^* =  \sup_{\tau \in \Psi_2} \inf_{\sigma \in  \Psi_1} \Upsilon_{\sigma, \tau}(\mu,t^*)$ is 1, i.e., if for every $\epsilon >0$ there is a strategy
$\tau_{\epsilon}$ of player 2 that guarantees that the
probability of reaching the target type is at least $1-\epsilon$ regardless of the strategy $\sigma$ of player 1; such a strategy $\tau_{\epsilon}$ is called $\epsilon -optimal$.
The main result of the paper is to provide polynomial-time algorithms
for both versions of the problem. The algorithms are nontrivial,
building upon the algorithms of both \cite{deAHK07} and \cite{ESY-icalp15-IC}
which both address different special subcases of qualitative BCSG reachability.

In the almost-sure problem, if the answer is positive, our algorithm
constructs (a compact description of) a strategy $\tau^*$ of player 2
that achieves value 1; the strategy is a randomized non-static
strategy, and this is inherent (i.e., there may not exist a static
strategy that achieves value 1). If the answer is negative, then our
algorithm constructs a (non-static, randomized) strategy $\sigma$ for
the opposing player 1 such that $\Upsilon_{\sigma, \tau}(\mu,t^*) < 1$
for all strategies $\tau$ of player 2.  In the limit-sure problem, if
the answer is positive, i.e., the value is 1, our algorithm constructs
for any given $\epsilon >0$, a static, randomized $\epsilon$-optimal
strategy, i.e., a strategy $\tau_{\epsilon}$ such that
$\Upsilon_{\sigma, \tau_{\epsilon}}(\mu,t^*) \geq 1-\epsilon$ for all
$\sigma \in \Psi_1$.  If the answer is negative, i.e., the value is 
$<
1$, our algorithm constructs a static randomized strategy $\sigma'$ for
player 1 such that $\sup_{\tau \in \Psi_2} \Upsilon_{\sigma',\tau} <
1$.

\medskip

\noindent {\bf Related Work.}
As mentioned, the two works most closely related to ours are \cite{deAHK07}
and \cite{ESY-icalp15-IC}.  Our results generalize both.
Firstly, de Alfaro,
Henzinger, and Kupferman \cite{deAHK07} studied finite-state
concurrent (stochastic) games (CSGs) 
with reachability objectives and provided polynomial
time algorithms for their qualitative analysis, both for the
almost-sure and the limit-sure reachability problem
(see also \cite{HKM09,FM13,HIM14,ChatHansenIbsen17} for more recent results on 
finite-state CSG reachability).  Branching Markov
Decision Processes (BMDPs) and Branching Simple Stochastic Games
(BSSGs) with reachability
objectives were studied in \cite{ESY-icalp15-IC}, which provided
polynomial-time algorithms for their qualitative analysis.  The
paper \cite{ESY-icalp15-IC}
also gave polynomial time algorithms for the approximate quantitative
analysis of BMDPs, i.e., for the approximate computation of the
optimal reachability probability for maximizing and minimizing BMDPs,
and showed that this problem for BSSGs is in TFNP. Note that even for
finite-state simple stochastic games the question of whether the value
of the game can be computed in polynomial time is a well-known
long-standing open problem \cite{Condon92}. It was also shown in
\cite{ESY-icalp15-IC} that the optimal non-reachability
probabilities of maximizing or minimizing BMDPs and BSSGs were
captured by the greatest fixed point of a system of equations
$x=P(x)$, where the right-hand side $P_i(x)$ of each equation is the
maximum or minimum of a set of probabilistic polynomials in $x$; note
that these types of equation systems are special cases of minimax-PPS,
and correspond to the case where in each one-shot game on the rhs of the
minimax-PPS equations only one of the two players has a choice of actions.

The quantitative problem for finite-state concurrent games, i.e.,
computing or approximating the value $v^*$ of the game (the optimal
reachability probability), has been studied previously and seems to be
considerably harder than the qualitative problem.  The problem of
determining if the value $v^*$ exceeds a given rational number, for
example 1/2, is at least as hard as the long-standing square-root sum
problem (\cite{rcsg2008}), a well-known open problem in numerical computation, which is
currently not known whether it is in NP or even in the polynomial
hierarchy. The problem of approximating the value
$v^*$ within a given desired precision can be solved however in the
polynomial hierarchy, specifically in TFNP[NP] \cite{FM13}.  It is
open whether the approximation problem is in NP (or moreover in P).
It was shown in \cite{HIM14} that the standard algorithms for
(approximately) solving these games, value iteration and policy
iteration, can be extremely slow in the worst-case: they can take a
doubly exponential number of iterations to obtain any nontrivial
approximation, even when the value $v^*$ is 1.  
Note also that there
are finite-state CSGs, with reachability value$=1$, 
for which (near-)optimal strategies 
for minimizer (maximizer, respectively) need to 
have some action
probabilities that are doubly-exponentially small 
\cite{HKM09,ChatHansenIbsen17}; thus a fixed
point representation of the probabilities would need an exponential
number of bits, and one must use a suitable compact
representation to ensure polynomial space. This is of course the case
also for branching stochastic games; the optimal or $\epsilon$-optimal
strategies constructed by our algorithms may use double-exponentially
small probabilities, which can however be represented succinctly so
that the algorithms run in polynomial time.

Another important objective, the probability of extinction, has been
studied previously for Branching Concurrent Stochastic Games, as well
as BMDPs and BSSGs, and the purely stochastic model of Branching
Processes (BPs). These branching models under the extinction objective
are equivalent to corresponding subclasses of recursive Markov models,
called respectively, 1-exit Recursive Concurrent Stochastic Games
(1-RCSG), Markov Decision Processes (1-RMDP), and Markov Chains
(1-RMC), and related subclasses of probabilistic pushdown processes
under a termination objective
\cite{rmc,ESY12,rmdp,esy-icalp12,rcsg2008,EKM}. The extinction
probabilities for these models are captured by the {\em least fixed
  point} (LFP) solutions of similar systems of probabilistic
polynomial equations; for example, the optimal extinction
probabilities of a BCSG are given by the LFP of a
minimax-PPS. Polynomial time-algorithms for qualitative analysis, as
well as for the approximate computation of the optimal extinction
probabilities of Branching MDPs (and 1-RMDPs) were given in
\cite{rmdp,esy-icalp12}.  However, negative results were shown also
which indicate that the problem is much harder for branching
concurrent (or even simple) stochastic games, even for the qualitative
extinction problem.  Specifically, it was shown in \cite{rmdp} that
the qualitative extinction (termination) problem for BSSG
(equivalently, 1-RSSG) is at least as hard as the well-known open
problem of computing the value of a
finite-state simple stochastic game \cite{Condon92}.  
Furthermore, it
was shown in \cite{rcsg2008} that (both 
the almost-sure and limit-sure) qualitative extinction problems
for BCSGs (equivalently 1-RCSGs) are at least as hard as the square-root
sum problem,  which is not even known to be in NP.\footnote{The 
results in \cite{rcsg2008} were phrased in terms of the limit-sure
problem, where it was shown that (a) deciding whether the
value of a finite-state CSG reachability game is at least
a given value $p \in (0,1)$ is square-root-sum-hard,
and (b) that the former problem is reducible to the limit-sure
decision problem for BCSG extinction games.
But the hardness proofs of (b) and (a) in \cite{rcsg2008} 
apply {\em mutatis mutandis} 
to (b) the almost-sure
problem for BCSG extinction, and to (a) the corresponding
problem of deciding, given a finite-state CSG and a value $p \in (0,1)$,
whether the maximizing
player has a strategy that achieves at least value $p$,
regardless of the strategy of the minimizer.
Thus, both the almost-sure and limit-sure extinction problem
for BCSGs are square-root-sum hard, and also both are at least as hard
as Condon's problem of computing
the exact value of a finite-state SSG reachability game.}
 Thus, the
extinction problem for BCSGs seems to be very different than the
reachability problem for BCSGs: obtaining analogous results for the extinction
problem of BCSGs to those of the present paper for reachability would
resolve two major open problems.

The equivalence between branching models (like e.g. BPs, BMDPs, BCSGs)
and recursive Markov models (like 1-RMC, 1-RMDP, 1-RCSG) with respect
to extinction does not hold for the reachability objective. For
example, almost-sure and limit-sure reachability coincide for a BMDP,
i.e., if the supremum probability of reaching the target is 1 then
there exists a strategy that ensures reachability with probability
1. However, this is not the case for 1-RMDPs. Furthermore, it is known
that almost-sure reachability for 1-RMDPs can be decided in polynomial
time \cite{BBFK08,BBKO11}, but limit-sure reachability for 1-RMDPs is not
even known to be decidable. The qualitative reachability problem for
1-RMDPs and 1-RSSGs (and equivalent probabilistic pushdown models) was
studied in \cite{BBKO11, BKL14}. These results do not apply to the
corresponding branching models (BMDP, BSSG).
Another objective
considered in prior work is the {\em expected total reward} objective
for 1-RSSGs and (\cite{EWY08}) and 1-RCSGs (\cite{wojt2013}) with
positive rewards.  In particular, \cite{wojt2013} shows that the
``qualitative'' problem of determining whether the game value for a
1-RCSG total reward game is $= \infty$ is in PSPACE.  None of these
prior results have any implications for BCSGs with reachability
objectives.

For richer objectives beyond reachability or extinction,
Chen et. al. \cite{CDK12} studied model checking of 
purely stochastic branching
processes (BPs) with respect to properties expressed by
deterministic parity tree automata, and showed that the qualitative
problem is in P-time (hence this holds in particular for reachability
probability in BPs), and that the quantitative problem of comparing the
probability with a rational is in PSPACE. 
Michalevski and Mio \cite{MichMio15} extended this
to properties of BPs  expressed by ``game automata'', a subclass
of alternating parity tree automata. 
More recently, Przyby{\l}ko and Skrzypczak \cite{PrzSkrz16}
considered existence and complexity of game values of Branching
turn-based (i.e., simple) stochastic games, with regular objectives, 
where the two players aim to maximize/minimize the probability that
the generated labeled tree belongs to a 
regular language (given by a tree automaton).
They showed that (unlike our case of simpler reachability games) already 
for some basic regular properties these games are not even
determined, meaning they do not have a value.
They furthermore showed that 
for a probabilistic turn-based branching
game, with a regular tree objective, it is undecidable
to compare the value that a given
player can force to $1/2$; 
whereas for deterministic turn-based branching games they showed 
it is decidable and \mbox{2-EXPTIME}-complete (respectively, EXPTIME-complete), 
to determine whether the player aiming to satisfy (respectively, falsify)
a given regular tree objective has a pure winning strategy. 
Other past research includes work in operations research on
(one-player) Branching MDPs \cite{pliska76,rotwhit82,denrot05a}. 
None of these prior works bear on any of the results on
BCSG reachability problems established in this paper.

\medskip

\noindent {\bf On the complexity of quantitative problems for BCSGs.}

All quantitative decision and approximation problems for BCSG
extinction and reachability games are in PSPACE.  
This follows by
exploiting the minimax-PPS equations whose least (and greatest) fixed
point solution captures the extinction (and non-reachability) values
of these games, and by then appealing to PSPACE upper bounds for deciding
the existential (and bounded-alternation) theory of
reals (\cite{Ren92}), in order to decide
questions about, and to approximate, the LFP and GFP of such equations.  
This was shown already for BCSG extinction
games in \cite{rcsg2008}.  A directly analogous proof yields
the same PSPACE upper bound for BCSG reachability games.  As mentioned before, 
the corresponding decision problems (e.g., deciding whether
the BCSG game value is at least a given probability $p \in (0,1)$),
are square-root-sum-hard,  already for finite-state CSG
reachability games \cite{rcsg2008}  (which are subsumed by
both BCSG extinction and BCSG reachability games).  This implies that 
even placing these decision problems in the
polynomial time hierarchy would require a breakthrough.  An
interesting question is how much the PSPACE upper bounds can be
improved for the {\em approximation} problems.  As noted
    earlier, Frederiksen and Miltersen \cite{FM13} have shown that for
    finite-state CSG reachability games, the game value can be
    approximated to desired precision in TFNP[NP].
We do not know an analogous complexity result for 
quantitative approximation
problems for BCSG extinction or reachability games,
nor do we know square-root-sum-hardness for these approximation problems.
We leave these as interesting open questions.

\medskip

\noindent{\bf Organization of the paper.}

Section 2 gives background and basic definitions. Section 3 shows the
relationship between the optimal non-reachability probabilities of a
game and the greatest fixed point of a minimax-PPS.  Section 4
presents the algorithm for determining if the value of a game is
0. Section 5 presents the algorithm for almost-sure reachability, and
Section 6 for limit-sure reachability.

\section{Background}

This section introduces some definitions and background
for Branching Concurrent Stochastic Games.  
It builds directly on, and generalizes, the 
definitions in \cite{ESY-icalp15-IC} associated with 
reachability problems for Branching MDPs and Branching Simple
Stochastic Games.

We first define the general model of a (multi-type) 
Branching Concurrent Stochastic Games(BCSGs), as well as
some important restrictions of the general model: 
Branching Simple Stochastic Games (BSSGs),
Branching MDPs (BMDPs), and (multi-type) Branching Processes (BPs).

\begin{definition}
  A \textbf{Branching Concurrent Stochastic Game(BCSG)} is a 2-player
  zero-sum game that consists of a finite set $V = \{T_1, \dots T_n\}$
  of types, two finite non-empty sets $\Gamma_{max}^i , \Gamma_{min}^i
  \subseteq \Sigma$ of actions (one for each player) for each type
  $T_i$ ($\Sigma$ is a finite action alphabet), and a finite set
  $R(T_i, a_{max}, a_{min})$ of probabilistic rules associated with
  each tuple $(T_i, a_{max}, a_{min})$, $i \in [n]$, where $a_{max}
  \in \Gamma_{max}^i$ and  $a_{min} \in \Gamma_{min}^i$. Each rule $r \in
  R(T_i, a_{max}, a_{min})$ is a triple $(T_i, p_r, \alpha_r)$, which
  we can denote by $T_i \xrightarrow{p_r} \alpha_r$, where $\alpha_r
  \in \mathbb{N}^n$ is a $n$-vector of natural numbers that
  denotes a finite multi-set over the set $V$, and where $p_r \in
  (0,1] \cap \mathbb{Q}$ is the probability of the rule $r$ (which we
  assume to be a rational number, for computational purposes), where
  we assume that for all $T_i \in V$ and $a_{max} \in \Gamma_{max}^i,
  \; a_{min} \in \Gamma_{min}^i$, the rule probabilities in $R(T_i,
  a_{max}, a_{min})$ sum to 1, i.e., $\sum_{r \in R(T_i, a_{max},
    a_{min})}p_r = 1$.
	\label{def:BCSG}
\end{definition}

If for all types $T_i \in V$,  either $|\Gamma_{max}^i| = 1$
or $|\Gamma_{min}^i| = 1$, then the model is a ``turn-based'' 
perfect-information game and is called a  \textbf{Branching
Simple Stochastic Game}  (\textbf{BSSG}).
If for all $T_i \in V$, $|\Gamma_{max}^i|
= 1$ (respectively, $|\Gamma_{min}^i| = 1$), then it is called a
\textit{minimizing} \textbf{Branching Markov Decision Process} (\textbf{BMDP})
(respectively, a \textit{maximizing} BMDP). If both
$|\Gamma_{min}^i| = 1 = |\Gamma_{max}^i|$ for all $i \in [n]$, then
the process is a classic, purely stochastic, \textbf{multi-type Branching
Process (BP)}  (\cite{Harris63}).

A {\em play} of a BCSG  defines a (possibly infinite) 
node-labeled forest, whose nodes are labeled by the type of the
object they represent.
A play contains a sequence of ``generations'', $X_0, X_1, X_2,
\dots$ (one for each integer time $t \ge 0$, corresponding
to nodes at depth/level $t$ in the forest).
For each $t \in {\mathbb{N}}$,  $X_t$ consists of the 
population (set of objects of given types), 
at time $t$.  
$X_0$ is
the initial population 
at generation
0 (these are the roots of the forest).  
$X_{k+1}$ is obtained from $X_k$ in the following way: for
each object $e$ in the set $X_k$, assuming $e$ has type $T_i$, 
both players select simultaneously
and independently actions 
$a_{max} \in \Gamma_{max}^i,$ and  
$a_{min} \in \Gamma_{min}^i$   (or distributions on such actions),   
according to their strategies; 
thereafter a rule $r \in R(T_i,a_{max}, a_{min})$ is chosen randomly 
and independently (for object $e$) with
probability $p_r$; each such object $e$ in $X_k$ is then replaced by 
the set of
objects specified by the multi-set $\alpha_r$ associated 
with the corresponding randomly chosen rule $r$. This process 
is repeated in each generation, 
as long as the current generation is not empty, and if for some $k \ge
0, \; X_k = \emptyset$ then we say the process \textit{terminates} or
becomes \textit{extinct}.

The strategies of the players can in general be arbitrary.
Specifically, at each generation, $k$, each player can,
in principle, select actions for the objects in $X_k$ 
based on the entire past history, may use randomization (a
mixed strategy), and may make different choices for objects of the
same type.  The  {\em history} of the process up to time $k-1$ 
is a forest of depth $k-1$ that
includes
not only the populations $X_0, X_1, \ldots, X_{k-1}$, but also the
information regarding all the past actions and rules applied and the
parent-child relationships between all the objects up to the
generation of $k-1$.  The history can be represented by a forest of
depth $k-1$, with internal nodes labelled by rules and actions, and
whose leaves at level $k-1$ form the population $X_{k-1}$.  Thus, a
strategy of player 1 (player 2, respectively) 
is a function that maps every 
finite history
(i.e., labelled forest of some finite depth as above) to a 
function that maps each object $e$ in the current population $X_k$
(assuming that the history has depth $k$) to a 
probability distribution 
on the actions $\Gamma^i_{max}$  (to the actions $\Gamma^i_{min}$, respectively),
assuming that object $e$ has type $T_i$.

Let $\Psi_1, \Psi_2$ be the set of all
strategies of players 1, 2.  We say that a strategy is {\em
  deterministic} if for every history it maps each object $e$ 
in the current population to a 
single action
with probability 1  (in other words, it does not randomize
on actions).  We say that a strategy is {\em static} if for
each type $T_i \in V$, 
and for any object $e$ of type $T_i$, the player always chooses the
same distribution on actions, irrespective of the history.

Different objectives can be considered for the BCSG game model. The
\textit{extinction} (or \textit{termination}) objective,
where players aim to maximize/minimize 
the extinction probability, has
already been studied in detail in \cite{rmdp} for BSSGs and in \cite{rcsg2008}
for BCSGs\footnote{Strictly speaking,
the model studied in \cite{rcsg2008} is {\em 1-exit Recursive concurrent
stochastic games} (1-RCSGs) with the objective of {\em termination}, 
but such games are easily seen to be equivalent to BCSGs
with the extinction objective: there is a simple linear-time transformation
from a 1-RCSG termination game to a BCSG extinction game, and vice versa
(see \cite{rmc} for the same correspondence, in the purely stochastic
 setting).}.  In particular, in \cite{rcsg2008} it 
was shown that the player minimizing extinction probability 
for BCSGs always has
an optimal (randomized) static strategy,
whereas the player maximizing extinction probability in general may only have
$\epsilon$-optimal randomized static strategies, for all $\epsilon > 0$.
(For BSSGs, it was shown in \cite{rmdp} that
both players have optimal deterministic static strategies
for optimizing extinction probability.)

This paper, on the other hand, deals with the (existential)
\textit{reachability} objective for BCSGs, where the aim of the players is to
maximize/minimize the probability of reaching a generation that
contains at least one object of a given target type $T_{f^*}$.
This objective was previously studied in \cite{ESY-icalp15-IC},
but only for BMDPs and BSSGs, not for the more general model of BCSGs.
It was already shown in
\cite{ESY-icalp15-IC} that in a BSSG
the player minimizing reachability probability
always has a deterministic static optimal strategy,
whereas
(unlike for the extinction objective) 
in general there need not 
exist {\em any} optimal strategy for the player maximizing 
reachability probability in a BMDP 
(and hence also in a BSSG and BCSG).
On the other hand, it was shown 
in \cite{ESY-icalp15-IC} that for BMDPs and BSSGs,
if the reachability game value is $=1$, then there is in
fact an optimal strategy (but not in general a static one,
even when randomization is allowed)
for the player maximizing the
reachability probability that forces the value $1$
(irrespective of the strategy of the player minimizing
the reachability probability).  It was
also shown that deciding whether the value $=1$ 
for BSSG reachability game can be decided
in P-time, and if the answer is ``yes'' then an optimal (non-static,
but deterministic)
strategy that achieves reachability value $1$ for the maximizer 
can be computed in P-time,
whereas if the answer is ``no'' a deterministic static strategy
that forces value $< 1$ can be computed for the minimizer 
in P-time.

We will show in this paper that 
the reachability game also has a \textit{value} for 
the more general imperfect-information game class of BCSGs. 
We do so by establishing systems of nonlinear 
minimax-equations whose greatest fixed point gives the
vector of values of the non-reachability game.

Let us note right away that there is a natural ``duality''
between the objectives of optimizing reachability probability and that of
optimizing extinction probability for BCSGs.    This duality was previously
detailed in \cite{ESY-icalp15-IC} for BSSGs.
The objective of optimizing the extinction probability (i.e., the
probability of generating a finite tree), starting from a
given type, can equivalently be rephrased as a 
``{\em universal reachability}'' objective (on a slightly
modified BCSG), 
where the goal is to
optimize the probability of eventually reaching the 
target type (namely ``death'')
on {\em all}  paths starting at the root of the tree. 
Likewise, the ``universal reachability'' objective 
can equivalently be rephrased as the objective of optimizing 
extinction probability (on a slightly modified BCSG).
By contrast, the
{\em reachability} objective that we study in this paper is the
``{\em existential reachability}'' objective 
of optimizing
the probability of reaching the target type on {\em some} path in the
generated tree.  
Despite this natural duality between these two objectives,
we show that there is a wide disparity between them,
both in terms of the nature and existence of optimal strategies, 
and in terms
of computational complexity: we show that the qualitative 
(existential) reachability problem for BCSGs can be solved in 
polynomial time, both in the almost-sure and limit-sure sense.

The BCSG reachability game can of course also be viewed 
as a ``non-reachability'' game (by just reversing the
role of the players).   It turns out this is useful to do,
and we will exploit it in crucial ways (and this was also
exploited in \cite{ESY-icalp15-IC} for BMDPs and BSSGs).
So we provide some notation for this purpose. 
Given an initial population $\mu \in \mathbb{N}^n$, with
$\mu_{f^*} = 0$, and given an integer $k \ge 0$, and strategies
$\sigma \in \Psi_1, \tau \in \Psi_2$, let $g_{\sigma, \tau}^k(\mu)$ be
the probability that the process does \textit{not} reach a generation
with an object of type $T_{f^*}$ in at most $k$ steps, under
strategies $\sigma, \tau$ and starting from the initial population
$\mu$. To be more formal, this is the probability that $(X_l)_{f^*} =
0$ for all $0 \le l \le k$. Similarly, let $g_{\sigma, \tau}^*(\mu)$
be the probability that $(X_l)_{f^*} = 0$ for all $l \ge 0$. We define
$g^k(\mu) = \sup_{\sigma \in \Psi_1} \inf_{\tau \in \Psi_2} g_{\sigma,
  \tau}^k(\mu)$ to be the value of the $k$-step non-reachability game
for the initial population $\mu$, and $g^*(\mu) = \sup_{\sigma \in
  \Psi_1} \inf_{\tau \in \Psi_2} g_{\sigma, \tau}^*(\mu)$ to be the
\textit{value} of the game under the non-reachability objective and
for the initial population $\mu$.  The next section will demonstrate that
these games are determined, meaning they have a value where $g^*(\mu)
= \sup_{\sigma \in \Psi_1} \inf_{\tau \in \Psi_2} g_{\sigma,
  \tau}^*(\mu) = \inf_{\tau \in \Psi_2} \sup_{\sigma \in \Psi_1}
g_{\sigma, \tau}^*(\mu)$. Similarly, for $g^k(\mu)$.

In the case where the initial population $\mu$ is a single object of
some given type $T_i$, then for the value of the game we write $g_i^*$
(or similarly, $g_i^k$, and when strategy $\sigma$ and $\tau$ are fixed,
we write $(g^*_{\sigma,\tau})_i$).
The collection of these values, namely the
vector $g^*$ of $g_i^*$'s, is called the vector of the
non-reachability values of the game. We will see that, having the
vector of $g_i^*$'s, the non-reachability value for a starting
population $\mu$ can be computed simply as $g^*(\mu) = f(g^*, \mu) :=
\prod_i (g_i^*)^{\mu_i}$. So given a BCSG, the \textit{aim} is to
compute the vector of non-reachability values. As our original
objective is reachability, we point out that the vector
of reachability values is $r^* = \mathbf{1} - g^*$ (where $\mathbf{1}$ is
the all-1 vector), and hence 
the reachability value $r^*(\mu)$ of the game starting
with population $\mu$ is 
$r^*(\mu) = 1 - g^*(\mu)$.

We will associate with any given BCSG a system of \textit{minimax
probabilistic polynomial equations} (\textbf{minimax-PPS}), $x = P(x)$, 
for the non-reachability
objective. This system will be constructed to have 
one variable $x_i$ and one equation $x_i = P_i(x)$ for each
type $T_i$ other than the target type $T_{f^*}$.
We will show that the vector of
non-reachability values $g^*$ for different starting types is
precisely the Greatest Fixed Point(GFP) solution of the system $x =
P(x)$ in $[0,1]^n$.

In order to define these systems of equations, some shorthand
notation will be useful. We use $x^v$ to denote the monomial $x_1^{v_1} *
x_2^{v_2} \cdots * x_n^{v_n} $ for an $n$-vector of variables $x =
(x_1, \cdots, x_n)$ and a vector $v \in \mathbb{N}^n$. Considering a
multi-variate polynomial $P_i(x) = \sum_{r \in R} p_r x^{\alpha_r}$
for some rational coefficients $p_r, r \in R$, we will call $P_i(x)$ a
\textbf{probabilistic polynomial}, if $p_r \ge 0$ for all $r \in R$
and $\sum_{r \in R}p_r \le 1$.
\begin{definition} 
  A \textbf{probabilistic polynomial system of equations} (\textbf{PPS}), $x = P(x)$,
  is a system of $n$ equations, $x_i = P_i(x)$,
  in $n$ variables where for all $i \in \{1, \dots, n\}$, $P_i(x)$ is
  a probabilistic polynomial.

  A \textbf{minimax probabilistic polynomial system of equations} (\textbf{minimax-PPS}), $x
  = P(x)$, is a system of $n$ equations
  in $n$ variables $x = (x_1,\dots,x_n)$, where for each $i \in
  \{1,\dots,n\}$,  $P_i(x) :=
  Val(A_i(x))$ is an associated \textsc{Minimax-probabilistic-polynomial}.  
By this we mean that $P_i(x)$ is
defined to be, for each $x \in \real^n$, the minimax value of the two-player zero-sum
matrix game given by a finite game payoff matrix $A_i(x)$ whose rows are
indexed by the actions $\Gamma_{max}^i$, and whose columns
are indexed by the actions $\Gamma_{min}^i$, where, 
for each pair $a_{max} \in \Gamma^i_{max}$ and $a_{min} \in \Gamma^i_{min}$,
the matrix entry $(A_i(x))_{a_{max}, a_{min}}$ is
given by a probabilistic polynomial $q_{i,a_{max},a_{min}}(x)$.
Thus, if $n_i = |\Gamma_{max}^i|$ and
$m_i = |\Gamma_{min}^i|$, and if we assume
w.l.o.g. that $\Gamma_{max}^i = \{1,\ldots,n_i\}$ and
that $\Gamma_{min}^i = \{1,\ldots,m_i\}$, then
$Val(A_i(x))$
 is defined as the minimax value of the zero-sum matrix game,
given by the following payoff matrix:
\[A_i(x) =
\begin{bmatrix}
    q_{i,1,1}(x) & q_{i,1,2}(x) & \dots & q_{i,1,m_i}(x) \\
    q_{i,2,1}(x) & \hdotsfor{3} \\
    \vdots & \vdots & \vdots & \vdots \\
    q_{i,n_i,1}(x) & \hdotsfor{2} & q_{i,n_i,m_i}(x)
\end{bmatrix}
\]
with each $q_{i,j,k}(x) := \sum_{r \in R(T_i, j, k)} p_r x^{\alpha_r}$ 
being a probabilistic polynomial for the actions pair $j, k$. 

If for all $i \in \{1, \dots, n\}$, either $|\Gamma_{min}^i| = 1$ or $|\Gamma_{max}^i| = 1$, then we call such a system \textbf{min-max-PPS}.
If for all $i \in \{1, \dots, n\}$, $|\Gamma_{min}^i| = 1$ (respectively, if $|\Gamma_{max}^i| = 1$ for all $i$) then
we will call such a system a \textbf{maxPPS} (respectively, a \textbf{minPPS}). Finally, a \textbf{PPS} is a minimax-PPS with 
both $|\Gamma_{min}^i| = 1 = |\Gamma_{max}^i|$ for every $i \in \{1, \cdots, n\}$.
	\label{def:(minimax-)PPS}
\end{definition}

For computational purposes, 
we assume that all coefficients are rational and that there
are no zero terms in the probabilistic
polynomials, and we assume the
coefficients and non-zero exponents of each term are given in binary.
We  denote by $|P|$ the total bit encoding length of a system $x =
P(x)$ under this representation.

This paper will examine minimax-PPSs.
Since $P(x)$ defines a monotone function
$P:[0,1]^n \rightarrow [0,1]^n$,  it follows by 
Tarski's theorem (\cite{tarski55})) that 
any such system has both a \textbf{Least Fixed Point (LFP)} 
solution $q^* \in [0,1]^n$,
and 
a \textbf{Greatest Fixed Point(GFP)} solution, $g^* \in
[0,1]^n$.
In other words, $q^* = P(q^*)$ and
 $g^* = P(g^*)$ and
moreover, for any $s^* \in [0,1]^n$ such that $s^* = P(s^*)$,
we have $q^* \leq s^* \leq g^*$  (coordinate-wise inequality).

We will show that the GFP of a
minimax-PPS, $g^*$, corresponds to the vector of \textit{values} for
a corresponding BCSG with  \textit{non}-reachability objective. 
We note that it has previously been
shown in \cite{rcsg2008} that the LFP solution, $q^* \in [0,1]^n$
of a minimax-PPS is the vector of \textit{extinction/termination}
values for a corresponding (but different) BCSG with the extinction objective,
and that the GFP of a
min-max-PPS is the vector of \textit{non}-reachability values
for a corresponding BSSG \cite{ESY-icalp15-IC}.
\begin{definition}
	
  A (possibly randomized) \textbf{policy} for the max (min) player in
  a minimax-PPS, $x = P(x)$, is a function that assigns a probability
  distribution to each variable $x_i$ such that the support of the
  distribution is a subset of $\Gamma_{max}^i$
  \; ($\Gamma_{min}^i$, respectively), 
where these now denote the possible actions(i.e., choices of
rows and columns) available for the respective player in the game matrix
$A_i(x)$ that defines $P_i(x)$.
	\label{def:policy}
\end{definition}
\noindent Intuitively, a policy is the same as a static strategy in the 
corresponding BCSG.

\begin{definition}
  For a minimax-PPS, $x = P(x)$, and policies $\sigma$ and $\tau$ for
  the max and min players, respectively, we write $x = P_{\sigma,
    \tau}(x)$ for the PPS obtained by fixing both these policies. We
  write $x = P_{\sigma, *}(x)$ for the minPPS obtained by fixing
  $\sigma$ for the max player, and $x = P_{*, \tau}(x)$ for the maxPPS
  obtained by fixing $\tau$ for the min player. More specifically,
  for policy $\sigma$ for the max player, we define the
  minPPS, $x = P_{\sigma, *}(x)$, as follows: for all  $i \in [n]$,
  $(P_{\sigma, *}(x))_i := \min\{s_{k}: k \in \Gamma_{min}^i \}$, where
  $s_k := \sum_{j \in \Gamma_{max}^i} \sigma(x_i, j) * q_{i,j,k}(x)$, where
  $\sigma(x_i, j)$ is the probability that the fixed policy
  $\sigma$ assigns to action $j \in \Gamma_{max}^i$ in variable
  $x_i$. We similarly define $x = P_{*, \tau}(x)$ and $x = P_{\sigma,
    \tau}(x)$.

For a minimax-PPS, $x = P(x)$, and a (possibly randomized) policy, 
$\sigma$ for the max player, we use $q_{\sigma, *}^*$ and $g_{\sigma, *}^*$ 
to denote the LFP and GFP solution vectors of the corresponding minPPS, 
$x = P_{\sigma, *}(x)$, respectively. Likewise we use $q_{*, \tau}^*$ and $g_{*, \tau}^*$ to denote the LFP and GFP solution vectors of the maxPPS, 
$x = P_{*, \tau}(x)$.
	\label{def:fixingPolicy}
\end{definition}

{\bf Note:}  we  overload 
notations such as $(g^*_{\sigma, *})_i$  and $(g^*_{*,\tau})_i$
to mean slightly different things, depending on whether $\sigma$ and $\tau$ are 
as static strategies (policies),  or are more general non-static strategies.
Specifically, let $E_i \in \nat^n$ denote the unit vector
which is  $1$ in the $i$'th coordinate and $0$ elsewhere. 
When $\tau \in \Psi_2$ is a general non-static strategy we use the notation 
$(g^*_{*,\tau})_i := 
g^*_{*,\tau}(E_i) =  \sup_{\sigma \in \Psi_1} g^*_{\sigma,\tau}(E_i)$.    We likewise
define $(g^*_{\sigma,*})_i$.    It will typically be clear from the context
which interpretation of $(g^*_{*,\tau})_i$ is intended.

\begin{definition}
	For a minimax-PPS, $x = P(x)$, a policy $\sigma^*$ is called \textbf{optimal} for the max player for the LFP (respectively, the GFP) if $q_{\sigma^*, *}^* = q^*$ (respectively, $g_{\sigma^*, *}^* = g^*$). 
	
An optimal policy $\tau^*$ for the min player for the LFP and GFP, respectively, is defined similarly. 

For $\epsilon > 0$, a policy $\sigma'$ for the max player is called
$\epsilon$\textbf{-optimal} for the LFP (respectively, the GFP), if
$||q_{\sigma', *}^* - q^*||_{\infty} \le \epsilon$ (respectively,
$||g_{\sigma', *}^* - g^*||_{\infty} \le \epsilon$). An
$\epsilon$-optimal policy $\tau'$ for the min player is defined
similarly.
	\label{def:(epsilon-)optimal}
\end{definition}

For convenience in proofs throughout the paper and to simplify the
structure of the matrices involved in the
\textit{minimax-probabilistic-polynomials}, $P_i(x)$, we shall observe that minimax-PPSs 
can always be cast in 
the following normal form.

\begin{definition}
  A minimax-PPS in \textbf{simple normal form(SNF)}, $x = P(x)$, is a
  system of $n$ equations in $n$ variables $\{x_1, \cdots, x_n\}$,
  where each $P_i(x)$ for $i = 1, 2, \dots, n$ is one of three forms:
	\begin{itemize}
        \item \textsc{Form L:} $P_i(x) = a_{i,0} +
          \sum_{j=1}^{n}a_{i,j}x_j$, where for all $j$, $a_{i,j} \ge
          0$, and $\sum_{j=0}^{n}a_{i,j} \le 1$
		\item \textsc{Form Q:} $P_i(x) = x_jx_k$ for some $j,k$
		\item \textsc{Form M:} $P_i(x) = Val(A_i(x))$,
where $A_i(x)$ is a $(n_i \times m_i)$ matrix,
such that for all $a_{max} \in [n_i]$ and $a_{min} \in [m_i]$,
the entry $A_i(x)_{(a_{max}, a_{min})}  \in \{x_1,\ldots, x_n\}  \cup \{1\}$.

(The reason we also allow ``1'' as an entry in the matrices $A_i(x)$ will 
 become clear later in the context of our algorithm.)

	\end{itemize}

	\label{def:SNF-form}
\end{definition}

We shall often assume a minimax-PPS in its SNF form, and say that a 
variable $x_i$ is ``of form/type'' L, Q, or M,
meaning that $P_i(x)$ has the corresponding form. 
The following proposition shows 
that we can efficiently convert any minimax-PPS into SNF form.

\begin{proposition}
  [cf. \cite{rmc,ESY12,esy-icalp12}] Every minimax-PPS, $x = P(x)$,
  can be transformed in P-time to an ``equivalent" minimax-PPS, $y =
  Q(y)$ in SNF form, such that $|Q| \in O(|P|)$. More precisely, the
  variables $x$ are a subset of the variables $y$, and both the LFP
  and GFP of $x = P(x)$ are, respectively, the projection of the LFP
  and GFP of $y = Q(y)$, onto the variables $x$, and furthermore an
  optimal policy (respectively, $\epsilon$-optimal policy) for the LFP
  (respectively, GFP) of $x = P(x)$ can be obtained in P-time from an
  optimal (respectively, $\epsilon$-optimal) policy for the LFP
  (respectively, GFP) of $y = Q(y)$.
	\label{prop:TranslationSNF-PPS}
\end{proposition}
\begin{proof}
We can easily convert, in P-time, any minimax-PPS into SNF form, using the following procedure.
\begin{itemize}
\item For each equation $x_i = P_i(x) := Val(A_i(x))$, for each
  probabilistic polynomial $q_{i,j,k}(x)$ on the right-hand-side that
  is not a variable, add a new variable $x_d$, replace $q_{i,j,k}(x)$
  with $x_d$ in $P_i(x)$, and add the new equation $x_d =
  q_{i,j,k}(x)$.

\item For each equation $x_i = P_i(x) = \sum_{j=1}^{m} p_j
  x^{\alpha_j}$, where $P_i(x)$ is a probabilistic polynomial that is
  not just a constant or a single monomial, replace every
  (non-constant) monomial $x^{\alpha_j}$ on the right-hand-side that
  is not a single variable by a new variable $x_{i_j}$ and add the
  equation $x_{i_j} = x^{\alpha_j}$.

\item For each variable $x_i$ that occurs in some polynomial with
  exponent higher than 1, introduce new variables $x_{i_1}, \dots ,
  x_{i_k}$ where $k$ is the logarithm of the highest exponent of $x_i$
  that occurs in $P(x)$, and add equations $x_{i_1} = x_i^2,\; x_{i_2}
  = x_{i_1}^2, \dots, x_{i_k} = x_{i_{k-1}}^2$. For every occurrence
  of a higher power $x_i^l,\; l > 1$, of $x_i$ in $P(x)$, if the
  binary representation of the exponent $l$ is $a_k \dots a_2a_1a_0$,
  then we replace $x_i^l$ by the product of the variables $x_{i_j}$
  such that the corresponding bit $a_j$ is 1, and $x_i$ if $a_0 =
  1$. After we perform this replacement for all the higher powers of
  all the variables, every polynomial of total degree $>$2 is just a
  product of variables.

\item If a polynomial $P_i(x) = x_{j_1} \dots x_{j_m}$ in the current
  system is the product of $m > 2$ variables, then add $m - 2$ new
  variables $x_{i_1}, \dots , x_{i_{m - 2}}$, set $P_i(x) =
  x_{j_1}x_{i_1}$, and add the equations $x_{i_1} = x_{j_2}x_{i_2},\;
  x_{i_2} = x_{j_3}x_{i_3}, \dots , x_{i_{m - 2}} = x_{j_{m -
      1}}x_{j_m}$.

\end{itemize}
Now all equations are of the form L, Q, or M. 

The above procedure allows us to convert any minimax-PPS into one in
SNF form by introducing $O(|P|)$ new variables and blowing up the size
of $P$ by a constant factor $O(1)$. It is clear that both the LFP and
the GFP of $x = P(x)$ arise as the projections of the LFP and GFP of
$y = Q(y)$ onto the $x$ variables. Furthermore, there is an obvious
(and easy to compute) bijection between policies for the resulting SNF
form minimax-PPS and the original minimax-PPS.
\end{proof} 

Thus from now on, and for the rest of this paper \textit{we may assume
  if needed, without loss of generality, that all minimax-PPSs are in
  SNF normal form}.
\begin{definition}
  The \textbf{dependency graph} of a minimax-PPS, $x = P(x)$, is a
  directed graph that has one node for each variable $x_i$, and
  contains an edge $(x_i, x_j)$ if $x_j$ appears in $P_i(x)$. The
  dependency graph of a BCSG has one node for each type, and contains
  an edge $(T_i, T_j)$ if there is a pair of actions $a_{max} \in
  \Gamma_{max}^i, a_{min} \in \Gamma_{min}^i$ and a rule $T_i
  \xrightarrow{p_r} \alpha_r$ in $R(T_i, a_{max}, a_{min})$ such that
  $T_j$ appears in $\alpha_r$.
\end{definition}

\section{Non-reachability values for BCSGs and the Greatest Fixed
  Point}

This section will show that for a given BCSG with a target type
$T_{f^*}$, a minimax-PPS, $x = P(x)$, can be constructed such that its
\textit{Greatest Fixed Point}(GFP) $g^* \in [0,1]^n$ is precisely the
vector $g^*$ of non-reachability values for the BCSG.

For simplicity, from now on let us call a
\textit{maximizer} (respectively, a \textit{minimizer}) the player that
aims to \textit{maximize}(respectively, \textit{minimize}) the
probability of \textit{not} reaching the target type. That is, we swap
the roles of the players for the benefit of less confusion in
analysing the minimax-PPS. While the players' goals in the game are
related to the objective of reachability, the equations we construct
will capture the optimal non-reachability values in the GFP of the
minimax-PPS.

For each type $T_i \not= T_{f^*}$, the minimax-PPS will have an
associated variable $x_i$ and an equation $x_i = P_i(x)$, and the
\textsc{Minimax-probabilistic-polynomial} $P_i(x)$ is built in the following
way. For each action $a_{max} \in \Gamma_{max}^i$ of the
maximizer (i.e., the player aiming to maximize the probability of
\textit{not} reaching the target) and action $a_{min} \in
\Gamma_{min}^i$ of the minimizer in $T_i$, let $R'(T_i, a_{max},
a_{min}) = \{r \in R(T_i, a_{max}, a_{min}) \;|\; (\alpha_r)_{f^*} =
0\}$ be the set of probabilistic rules $r$ for type $T_i$ and players'
action pair $(a_{max}, a_{min})$ that generate a multi-set $\alpha_r$
which does not contain an object of the target type. For each actions
pair for $T_i$, there is a probabilistic polynomial
$q_{i,a_{max},a_{min}} := \sum_{r \in R'(T_i, a_{max},
  a_{min})}p_rx^{\alpha_r}$. Observe that there is no need to include
rules where $\alpha_r$ contains an item of type $T_{f^*}$, because
then the term with monomial $x^{\alpha_r}$ will be 0. Now after a
polynomial is constructed for each pair of players' moves, we
construct $P_i(x)$ as the value of a zero-sum matrix game $A_i(x)$,
where the matrix is constructed as follows: (1) rows belong to the max
player in the minimax-PPS (i.e., the player trying to maximize the
non-reachability probability), and columns belong to the min player;
(2) for each row and column (i.e., pair of actions $(a_{max}, a_{min})$)
there is a corresponding probabilistic polynomial $q_{i,a_{max},
  a_{min}}(x)$ in the matrix entry $A_i(x)_{a_{max},a_{min}}$.

The following theorem captures the fact that the optimal
\textit{non-reachability values} $g^*$ in the BCSG correspond to the
\textit{Greatest Fixed Point(GFP)} of the minimax-PPS.

\begin{theorem}
	The non-reachability game values $g^* \in [0,1]^n$ of a BCSG 
reachability game exist, and
correspond to the Greatest Fixed Point(GFP) of the minimax-PPS, $x = P(x)$, in $[0,1]^n$. That is, $g^* = P(g^*)$, and for all other fixed points $g' = P(g')$ in $[0,1]^n$, it holds that $g' \le g^*$. Moreover, for an initial population $\mu$, the optimal non-reachability value is $g^*(\mu) = \prod_i(g_i^*)^{\mu_i}$ and the game is determined, i.e., $g^*(\mu) = \sup_{\sigma \in \Psi_1} \inf_{\tau \in \Psi_2} g_{\sigma, \tau}^*(\mu) = \inf_{\tau \in \Psi_2} \sup_{\sigma \in \Psi_1} g_{\sigma, \tau}^*(\mu)$.
Finally, the player maximizing non-reachability probability in the
BCSG has a (mixed) static optimal strategy. 
	\label{theorem:GFP-NonReach}
\end{theorem}
\begin{proof}
Note that $P:[0,1]^n \rightarrow [0,1]^n$ is a monotone operator, since all coefficients in all the polynomials $P_i(x)$ are non-negative, and for $x \le y$, where $x,y \in [0,1]^n$, it holds that 
$A_i(x) \le A_i(y)$  (entry-wise inequality) and thus 
$Val(A_i(x)) \le Val(A_i(y))$. Thus, $P_i(x) \le P_i(y)$. Let $x^0 = \mathbf{1}$ and $x^k = P(x^{k-1}) = P^k(\mathbf{1}), \; k > 0$ be the $k$-fold application of $P$ on the vector $\mathbf{1}$ (i.e., the all-1 vector). 
By induction on $k$ the sequence $x^k$ is monotonically non-increasing, i.e., $x^{k+1} \le x^k \le \mathbf{1}$.

By Tarski's theorem (\cite{tarski55}), $P(\cdot)$ has a Greatest Fixed Point (GFP)
$x^* \in [0,1]^n$.  The GFP is the limit of the monotone the sequence $x^k$,
i.e., $x^* = \lim_{k \rightarrow \infty} x^k$. 
To continue the proof, we 
need the following lemma.

\begin{lemma}
  For any initial non-empty population $\mu$, assuming it does not
  contain the target type $T_{f^*}$, and for any $k \ge 0$, the value
  of not reaching $T_{f^*}$ in $k$ steps is $g^k(\mu) = f(x^k, \mu) :=
  \prod_{i=1}^n(x_i^k)^{(\mu)_i}$. Also, there are strategies for the
  players, $\sigma^k \in \Psi_1$ and $\tau^k \in \Psi_2$, that achieve
  this value, that is $g^k(\mu) = \sup_{\sigma \in \Psi_1} g_{\sigma,
    \tau^k}^k(\mu) = \inf_{\tau \in \Psi_2} g_{\sigma^k,
    \tau}^k(\mu)$.
	\label{lemma:GFP-k-step-NonReach}
\end{lemma}
\begin{proof}
 
  Before we begin the proof, let us make a quick observation. 
For a fixed vector $x \in \real^n$, 
consider the zero-sum matrix game defined by
the payoff matrix $A_i(x)$ for player 1 (the row player).
Consider  fixed mixed strategies ${\mathbf{s}}_i$ and ${\mathbf{t}}_i$ 
for the row and column players in this matrix game.
Thus, ${\mathbf{s}}_i(a_{max})$   (${\mathbf{t}}_i(a_{min})$, respectively) 
defines the probability placed on action $a_{max}$ (on action $a_{min}$, respectively) 
in $\mathbf{s}_i$  (in $\mathbf{t}_i$, respectively).
The expected payoff to player 1 (the maximizing player), 
under these mixed strategies is:
\begin{eqnarray}
	&& \sum_{a_{max} \in \Gamma^i_1 ,a_{min} \in \Gamma^i_2} \mathbf{s}_i(a_{max}) 
\mathbf{t}_i(a_{min}) 
q_{i,a_{max},a_{min}}(x)   \nonumber \\ & = & 
\sum_{a_{max},a_{min}} \Big[\mathbf{s}_i(a_{max}) \mathbf{t}_i(a_{min}) 
\sum_{r \in R'(T_i,a_{max},a_{min})}p_rx^{\alpha_r} \Big] \nonumber \\
	&= & \sum_{a_{max},a_{min}} \;\; \sum_{r \in R'(T_i,a_{max},a_{min})} 
\mathbf{s}_i(a_{max}) \mathbf{t}_i(a_{min}) p_r x^{\alpha_r} = \sum_{r \in R'(T_i)} p'_r x^{\alpha_r}
	\label{eq:matrixVal}
\end{eqnarray}
where $R'(T_i)$ is the set of all probabilistic rules for type $T_i$; the newly defined probability $p'_r$ of a rule $r$ is equal to 
$\mathbf{s}_i(a_{max}) * \mathbf{t}_i(a_{min}) * p_r$ for the pair $(a_{max},a_{min})$ 
for which the rule $r$ is in $R'(T_i,a_{max},a_{min})$, and where 
$\alpha_r$ is the population that rule $r$ generates, meaning
rule $r$ is defined by $T_i \stackrel{p_r}{\rightarrow} \alpha_r$. 

Now let us prove the Lemma by induction on $k$. For the basis step,
clearly $g^0(\mu) = \mathbf{1}$, since the initial population does not
contain any objects of the target type. Moreover, $x^0 = \mathbf{1}$
and so $f(\mathbf{1}, \mu) = \mathbf{1}$.

For the inductive step, first we demonstrate that $g^k(\mu) \ge f(x^k,
\mu)$. Consider a strategy ${{\sigma}}^k := (\hat{\mathbf{s}}, \sigma^{k-1})$ for the max
player (i.e., the player aiming to maximize the non-reachability
probability), constructed in the following way. 
For all $i$, and for every object of
type $T_i$ in the initial population $\mu = X_0$, the max player 
chooses as a first step the minimax-optimal mixed strategy 
$\hat{\mathbf{s}}_i$ 
in the
zero-sum matrix game $A_i(x^{k-1})$ (which exists, due to the
minimax theorem).  The min player (player 2), as part of its strategy,
chooses some distributions on actions for all objects in the population
$X_0$ (independently of player 1), and then the rules are chosen 
according to the
resulting probabilities, forming the next generation $X_1$ at time
1. Thereafter, the max player acts according to an optimal
$(k-1)$-step strategy $\sigma^{k-1}$, starting from population $X_1$
($\sigma^{k-1}$ exists by the inductive assumption, and we will indeed
prove by induction that the thus defined $k$-step strategy 
$\sigma^k$ is optimal in the $k$-step game).
 Note that ${{\sigma}}^k$ can be mixed, and can also be 
non-static since the
action probabilities can depend on the generation and history.

Now let $\tau$ be any strategy for the min player.  In
the first step, $\tau$ chooses some
distributions on actions for each object in $X_0 = \mu$. 
After the choices of ${\sigma}^k$ and $\tau$ are made
in the first step, rules are picked
probabilistically and the population $X_1$ is generated. By the inductive
assumption, $g^{k-1}(X_1) = f(x^{k-1}, X_1)$, i.e., the value of not
reaching the target type in next $k-1$ steps, starting in population
$X_1$, is precisely $f(x^{k-1}, X_1)$. Therefore, the $k$-step probability
of not reaching the target,
starting in $\mu$, using strategies $\sigma^k$ and $\tau$,
is $g^k_{\sigma^k,\tau}(\mu) = \sum_{X_1}p(X_1)g^{k-1}_{\sigma^{k-1},\tau}(X_1) 
\geq \sum_{X_1}p(X_1)f(x^{k-1}, X_1)$, where the sum is over all possible
next-step populations $X_1$, and in each term $f(x^{k-1}, X_1)$ 
is multiplied by the
probability $p(X_1)$ of generating that particular population $X_1$.
The reason for the inequality is because, by optimality of
$\sigma^{k-1}$ for the max player in the $(k-1)$-step game,
we know that $g^{k-1}_{\sigma^{k-1},\tau}(X_1) \geq g^{k-1}(X_1) = f(x^{k-1},X_1)$.  

The sum $\sum_{X_1}p(X_1)f(x^{k-1}, X_1)$ can
be rewritten as a product of $|\mu|$ terms, one for each object in the
initial population $X_0$. Specifically, given $X_0$, 
let $L_{X_0,X_1}$ denote the set of  
all possible tuples of rules $(r_1,\ldots,r_{|X_0|})$, which 
associate to each object $e_j$ in the population $X_0$, 
a rule $r_j$ such that
if $e_j$ has type $T_i$, then $r_j \in R'(T_i)$ is a rule for type $T_i$,
and furthermore such that if we apply the rules $(r_1,\ldots, r_{|X_0|})$,
they generate multisets $\alpha_1, \ldots, \alpha_{|X_0|}$,
such that we obtain the population $X_1 = \bigcup \alpha_i$ from them.
 
Then for $X_0 = \mu$, we can rewrite $\sum_{X_1}p(X_1)f(x^{k-1}, X_1)$ as:

\begin{eqnarray*} 
\sum_{X_1}p(X_1)f(x^{k-1}, X_1) & = & \sum_{X_1} \sum_{(r_1,\ldots,r_{|X_0|}) \in L_{X_0,X_1}} 
(\prod_{j=1}^{|\mu|} p'_{r_j}) \cdot
 (\prod_{j=1}^{|\mu|}f(x^{k-1}, \alpha_{r_j}))\\
& =  & \prod_{j=1}^{|\mu|} \sum_{r_j}p'_{r_j}f(x^{k-1}, \alpha_{r_j}) \quad , 
\end{eqnarray*}
where $r_j$ ranges over all
rules that can be generated by the type of object $e_j$, 
and $p'_{r_j}$ is the probability of generating rule $r_j$
for object $e_j$ in the first step, under strategies ${{\sigma}}^k$ and 
$\tau$.
$\alpha_{r_j}$ is the population produced from $e_j$
under rule $r_j$.  Note that the term 
$\sum_{r_j}p'_{r_j}f(x^{k-1}, \alpha_{r_j})$ for an object $e_j$ 
of type $T_i$ has the same form as 
equation (\ref{eq:matrixVal}) above.  
This observation implies that, since the mixed strategy 
$\hat{\mathbf{s}}_i$ is
minimax-optimal in the zero-sum matrix game with matrix $A_i(x^{k-1})$,   
the term $\sum_{r_j}p'_{r_j}f(x^{k-1}, \alpha_{r_j})$
corresponding to each object $e_j$ of type $T_i$
is  $\geq Val(A_i(x^{k-1}))
= P_i(x^{k-1}) = x_i^k$. Hence, for
any strategy $\tau$ chosen the min player, starting with the objects in 
$\mu = X_0$, the
probability of not reaching the target type in next $k$ steps under
strategies ${\sigma}^k$ and $\tau$ is 
$g^k_{{\sigma}^k, \tau}(\mu) \ge \prod_{i=1}^{|\mu|}x_i^k = f(x^k, \mu)$. 
Therefore, the $k$-step non-reachability value is
$g^k(\mu) = \sup_{\sigma \in \Psi_1} \inf_{\tau \in \Psi_2} g_{\sigma,
  \tau}^k(\mu) \ge \inf_{\tau \in \Psi_2} g_{\sigma^k, \tau}^k(\mu)
\ge f(x^k, \mu)$.

Symmetrically we can prove the reverse inequality by using the other
player as an argument. That is, similarly let $\tau^k$ selects as a
first step for each object of type $T_i$ in the initial population
$\mu = X_0$ the (mixed) optimal strategy in the corresponding zero-sum
matrix game $A_i(x^{k-1})$ (exists by the minimax
theorem). Simultaneously and independently the max player chooses
moves for the objects, and then rules are picked in order to generate
population $X_1$. Afterwards, the min player acts according to an
optimal $k-1$-step strategy $\tau^{k-1}$ (which exists by the inductive
hypothesis). As before, $g^k(\mu)$ can be written as a product of
$|\mu|$ terms, where each term is $\sum_{r_j} p'_{r_j}
f(x^{k-1}, \alpha_{r_j})$. Again, by the
choice of $\tau^k$, it follows that the term for each object $e_j$ 
of type $T_i$ is at
most $Val(A_i(x^{k-1})) = P_i(x^{k-1}) = x_i^k$. Thus, showing that
$\sup_{\sigma \in \Psi_1} g_{\sigma, \tau^k}^k(\mu) \le f(x^k, \mu)$,
and $g^k(\mu) \le f(x^k, \mu)$. So, at the end $g^k(\mu) =
\sup_{\sigma \in \Psi_1} g_{\sigma, \tau^k}^k(\mu) = \inf_{\tau \in
  \Psi_2} g_{\sigma^k, \tau}^k(\mu) = f(x^k, \mu) = \prod_{i=1}^n
(x_i^k)^{(\mu)_i}$.     Note that the constructed strategy $\sigma^k$
(and $\tau^k$) is thus optimal for the player maximizing (respectively,
minimizing), the probability of not reachability the target
type in $k$ steps.
If the initial population consists of a single
object of type $T_i \not= T_{f^*}$, then the Lemma states that $g_i^k
= x_i^k$ for all $k \ge 0$. 
\end{proof}

Now we continue the proof of Theorem \ref{theorem:GFP-NonReach}.
We
show that the game is determined, i.e., $g^*(\mu) = \sup_{\sigma \in
  \Psi_1} \inf_{\tau \in \Psi_2} g_{\sigma, \tau}^*(\mu) = \inf_{\tau
  \in \Psi_2} \sup_{\sigma \in \Psi_1} g_{\sigma, \tau}^*(\mu)$, and
that the game value for the objective of not reaching $T_{f^*}$ 
is precisely 
$f(x^*, \mu)$, where $x^* = \lim_{k \rightarrow \infty}x^k \in
[0,1]^n$ is the GFP of the system $x = P(x)$, which exists by Tarski's
theorem.   As a special case, if the initial population $\mu$ is just a single
object of type $T_i \not= T_{f^*}$, we have $g_i^* = x_i^*$.

Since the sequence $x^k$ converges to $x^*$ monotonically from above 
(recall $x^0 =
\mathbf{1}$ and the sequence is monotonically non-increasing), then
$f(x^k, \mu)$ converges to $f(x^*, \mu)$ from above, i.e., for any
$\epsilon > 0$ there is a $k(\epsilon)$ where $f(x^*, \mu) \le
f(x^{k(\epsilon)}, \mu) < f(x^*, \mu) + \epsilon$. By Lemma
\ref{lemma:GFP-k-step-NonReach}, the min player strategy
$\tau^{k(\epsilon)}$ (as described in the Lemma) achieves the
$k(\epsilon)$-step value of the game, i.e., $\sup_{\sigma \in \Psi_1}
g_{\sigma, \tau^{k(\epsilon)}}^{k(\epsilon)}(\mu) = f(x^{k(\epsilon)},
\mu) < f(x^*, \mu) + \epsilon$. But for any strategy $\sigma$,
$g_{\sigma, \tau^{k(\epsilon)}}^*(\mu) \le g_{\sigma,
  \tau^{k(\epsilon)}}^{k(\epsilon)}(\mu)$, since the more steps the
game takes, the lower the probability of non-reachability is. So it
follows that $\sup_{\sigma \in \Psi_1} g_{\sigma,
  \tau^{k(\epsilon)}}^*(\mu) \le \sup_{\sigma \in \Psi_1} g_{\sigma,
  \tau^{k(\epsilon)}}^{k(\epsilon)}(\mu) < f(x^*, \mu) +
\epsilon$. And since it holds for every $\epsilon > 0$, then
$\inf_{\tau \in \Psi_2} \sup_{\sigma \in \Psi_1} g_{\sigma,
  \tau}^*(\mu) \le f(x^*, \mu)$. Thus, by standard facts, $g^*(\mu) =
\sup_{\sigma \in \Psi_1} \inf_{\tau \in \Psi_2} g_{\sigma,
  \tau}^*(\mu) \le \inf_{\tau \in \Psi_2} \sup_{\sigma \in \Psi_1}
g_{\sigma, \tau}^*(\mu) \le f(x^*, \mu)$.

To show the reverse inequality, namely $g^*(\mu) \ge f(x^*, \mu)$, let
$\sigma^*$ be the (mixed) static strategy for the max player (i.e., the
player aiming to maximize the probability of \textit{not} reaching the
target type), that for each object of type $T_i$ always selects the
(mixed) optimal strategy in the zero-sum matrix game $A_i(x^*)$
(which exists by the minimax theorem). Fixing $\sigma^*$, the BCSG becomes a
minimizing BMDP and the minimax-PPS, $x = P(x)$, becomes a minPPS, $x
= P'(x) = P_{\sigma^*, *}(x)$. In this new system of equations, for
every type $T_i$ (i.e., variable $x_i$), the function on the right-hand
side changes from $P_i(x) = Val(A_i(x))$ to $P'_i(x) = \min \{m_b : b
\in \Gamma_{min}^i\}$, where $m_b := \sum_{j \in \Gamma_{max}^i}
\sigma^*(x_i, j)*q_{i,j,b}(x)$. Hence, $P'(x) \le P(x)$ for all $x \in
[0,1]^n$. Thus, if we denote by $y^k, k \ge 0$ the vectors obtained from the
k-fold application of $P'(x)$ on the vector $\mathbf{1}$ (i.e., the all-1
vector), then $y^k \le x^k$ for all $k \ge 0$. So it follows that $y^*
\le x^*$, with $y^*$ and $x^*$ being the GFP of $x = P'(x)$ and $x =
P(x)$, respectively. But since the fixed strategy $\sigma^*$ is the 
optimal
strategy for the max player with respect to vector $x^*$ and achieves
the value $P_i(x^*) = Val(A_i(x^*))$ for all variables, $x^*$
must also be a fixed point of $x = P'(x)$ and hence $x^* = y^*$.

Now consider any strategy $\tau$ for the min player in the minimizing
BMDP. Recall that a minimizing BMDP is a BCSG where in every type the
max player has a single available action. Then by the induction step
in the proof of Lemma \ref{lemma:GFP-k-step-NonReach} it holds that
for every $k \ge 0$, starting in the initial population $\mu$, the
probability of \textit{not} reaching the target type $T_{f^*}$ in $k$
steps under strategy $\tau$ is at least $f(y^k, \mu)$. Hence, the
infimum probability of \textit{not} reaching the target type (in any
number of steps) is at least $\lim_{k \rightarrow \infty} f(y^k, \mu) =
f(y^*, \mu) = f(x^*, \mu)$. Therefore, $\inf_{\tau \in \Psi_2}
g_{\sigma^*, \tau}^*(\mu) \ge f(x^*, \mu)$. However, we know that
$g^*(\mu) = \sup_{\sigma \in \Psi_1} \inf_{\tau \in \Psi_2} g_{\sigma,
  \tau}^*(\mu) \ge \inf_{\tau \in \Psi_2} g_{\sigma^*, \tau}^*(\mu)$,
which shows the reverse inequality.

We can deduce that $g^*(\mu) = \sup_{\sigma \in \Psi_1} \inf_{\tau \in
  \Psi_2} g_{\sigma, \tau}^*(\mu) = \inf_{\tau \in \Psi_2}
\sup_{\sigma \in \Psi_1} g_{\sigma, \tau}^*(\mu) = f(x^*, \mu)
=\;\allowbreak \inf_{\tau \in \Psi_2} g_{\sigma^*, \tau}^*(\mu)$ and
$\sigma^*$ is an optimal (mixed) static strategy for the max player
under the non-reachability objective. 
\end{proof}

Note that the player minimizing the non-reachability probability
need not have any optimal strategy, even for a BMDP  (see Example 3.2 in
\cite{ESY-icalp15-IC}).

\begin{corollary}
Given a BCSG reachability game, and a probability $p \in (0,1)$, 
deciding whether the game value is $\geq p$ is in PSPACE.
\end{corollary}

The PSPACE upper bound follows from Theorem
\ref{theorem:GFP-NonReach}, by appealing to decision procedures
for the (existential) theory of reals to answer quantitative
questions about the GFP of the corresponding minimax-PPS equations.
This is entirely analogous
to very similar arguments in \cite{rcsg2008,ESY-icalp15-IC,rmc}, so we do
not elaborate.  Any substantial improvement on PSPACE for such
quantitative decision problems would require a major breakthrough on
exact numerical computation, even for BPs or BMDPs (see
\cite{rmc,ESY-icalp15-IC,rcsg2008}).

\section{P-time algorithm for deciding reachability value $= 0$ for BCSGs}

In this section we show that there is a P-time algorithm for computing
the variables $x_i$ with value $g_i^* = 1$ for the GFP in a given
minimax-PPS, or in other words, for a given BCSG, deciding whether the
value for reaching the target type, starting with an object of 
a given type $T_i$,
is $0$. The algorithm does not take into consideration the actual
probabilities on the transitions in the game (i.e., the coefficients
of the polynomials), but rather depends only on the structure of the game
(respectively, the dependency graph structure of the minimax-PPS) and
performs an AND-OR graph reachability analysis.  The algorithm is 
easy, and is very
similar to the algorithm given for deciding $g_i^* = 1$ for BSSGs in
\cite{ESY-icalp15-IC}.

\begin{proposition}{(cf. \cite{ESY-icalp15-IC}, Proposition 4.1)}
  There is a P-time algorithm that given a BCSG or equivalently a
 corresponding
  minimax-PPS, $x=P(x)$, with n variables and GFP $g^* \in [0,1]^n$,
  and given $i \in [n]$, 
  decides whether $g_i^* = 1$ or $g_i^* < 1$. Equivalently, for a
  given BCSG with non-reachability objective and a starting object of
  type $T_i$, it decides whether the non-reachability game value is 1. 
  In the case of $g_i^* =
  1$, the algorithm produces a deterministic policy (or deterministic
  static strategy in the BCSG case) $\sigma$ for the max player 
(maximizing non-reachability) that
  forces $g_i^* = 1$. Otherwise, if $g^*_i < 1$, 
the algorithm produces a mixed policy $\tau$ (a mixed static strategy) 
for the min player (minimizing non-reachability) that guarantees $g_i^* < 1$.
	\label{prop:QualNonReach}
\end{proposition}
\begin{proof}

Let $W = \{x_1,\ldots,x_n\}$ denote the set of all variables
in the minimax-PPS, $x=P(x)$.
  Recall that the dependency graph of $x=P(x)$ 
  has a directed edge $(x_i, x_j)$ iff
  variable $x_i$ depends on variable $x_j$, i.e., $x_j$ occurs in
  $P_i(x)$. Let us call a variable $x_i$  \textit{deficient} if
  $P_i(x)$ is of form L and $P_i(\mathbf{1}) < 1$. Let $Z \subseteq
 \{x_1, \ldots, x_n\}$ be the set
  of deficient variables. 
The remaining variables $X = W - Z$ are
  partitioned, according to their SNF-form equations: $X = L \cup Q
  \cup M$.

\begin{figure}[ht]
	\begin{enumerate}
		\item Initialize $S := Z$.
		\item Repeat until no change has occurred:
		\begin{enumerate}
			\item if there is a variable $x_i \not\in S$ of form L or Q such that $P_i(x)$ contains a variable already in $S$, then add $x_i$ to 
$S$.
			\item if there is a variable 
$x_i \not\in S$ of form M such that for every  action
$a_{max} \in \Gamma_{max}^i$, there exists an action
$a_{min} \in \Gamma_{min}^i$,  such that $A_i(x)_{(a_{max}, a_{min})} \in S$,
then add $x_i$ to $S$.
		\end{enumerate}
\item Output the set $\bar{S} := W - S$.
	\end{enumerate}
	\caption{Simple P-time algorithm for computing the set
of types with reachability value $0$ in a given BCSG, or equivalently the set
of variables $\{x_i \;|\; g_i^* = 1\}$ of the associated minimax-PPS.}
	\label{fig:QualNonReach}
\end{figure}

Figure \ref{fig:QualNonReach} gives the algorithm.  The intuition
behind it is as follows: notice that in 2.(b) no matter what
strategy the max player chooses in the particular variable (i.e., type
in the game), the min player can ensure with positive probability to
end up in a successor variable that already is bad for the max
player. The resulting winning strategies for players' corresponding
winning sets (it is irrelevant to define strategies in the losing
nodes) are: (i) for $x_i \in S$, the min player's strategy (mixed
static) $\tau$ selects uniformly at random among the ``witness'' moves
from step 2.(b), and (ii) for $x_i \in \overline{S}$ the max player's
strategy (deterministic static) $\sigma$ chooses an action $a_{max}
\in \Gamma_{max}^i$ that ensures staying within $\overline{S}$ no
matter what the minimizer's action (which must exist, otherwise $x_i$ would have
been added to $S$).

We need to prove that $g_i^* < 1$ iff $x_i \in
S$. First, we show that $x_i \in S$ implies $g_i^* < 1$.
Assume $x_i \in S$ (and therefore
$\tau$ is defined).  We analyse by induction, based on the time (iteration) 
in which 
variable $x_i$ was added to $S$ in the iterative algorithm. For the
base case, if $x_i$ was added at the initial step (i.e., $x_i
\in Z$), then $g_i^* \le P_i(\mathbf{1}) < 1$. For the
induction step, if variable $x_i$ is of type L or Q, then $g_i^* =
P_i(g^*)$ is a linear combination (with positive coefficients whose
sum is $\le 1$) or a quadratic term, containing at least one variable $x_j$
that was already in $S$ prior to $x_i$, and hence, by induction, $g_j^* < 1$. 
Hence, $g^*_i < 1$. If $x_i$ is
of form M, then for $\forall a_{max} \in \Gamma_{max}^i, \; \exists
a_{min} \in \Gamma_{min}^i$ such that the corresponding variable
$x_{(a_{max}, a_{min})} \in S$ (i.e., $g_{(a_{max},a_{min})}^* < 1$),
and $\tau$ gives positive probability to all such witnesses
$a_{min}$. So for any strategy $\sigma$ that the maximizer picks,
$\sum_{a_{min}, a_{max}} \sigma(x_i, a_{max})\tau(x_i,
a_{min})g_{(a_{max},a_{min})}^* < 1$. It follows that for any 
strategy $\sigma$, $(g_{\sigma,\tau}^*)_i < 1$, or in other words $(g_{*,
  \tau}^*)_i < 1$. Thus, $g_i^* \le (g_{*, \tau}^*)_i < 1$.

Next, to show that if $g^*_i < 1$ then $x_i \in S$, we show
the contrapositive.   Assume $x_i \in
\overline{S}$ (and therefore $\sigma$ is defined). All variables of form
$L \cup Q$ depend only on variables in $\overline{S}$ (otherwise they would
have been added to $S$). Moreover, for every $x_i$ of type M, there is
a maximizer action $a_{max}$ such that, 
all variables in row $a_{max}$ of the matrix
of $A_i(x)$ are in
$\overline{S}$. 
If no such action exists, then $x_i$ would have been
added to $S$ in step 2.(b). 
Let $\sigma(x_i)$ choose such an action $a_{max}$ deterministically 
(i.e., with probability 1).  
In the dependency graph of the resulting minPPS, $x=P_{\sigma,*}(x)$, 
there are no edges from $\bar{S}$ to $S$: all variables
of type $L$, $Q$, or $M$ depend only on $\bar{S}$ variables, otherwise
they would have been added to $S$.
Moreover, $\overline{S}$ does not contain
any deficient variables. So, $P_i(\mathbf{1}) = 1$ for every $x_i \in
\overline{S}$, and the all-1 vector is a fixed point for 
the subsystem of the minPPS,
$x = P_{\sigma,*}(x)$ induced by the variables $\bar{S}$.  In other words, 
$(g_{\sigma,*}^*)_i = 1$ (thus $g_i^* = 1$) for all $x_i \in \bar{S}$.
\end{proof}

\section{P-time algorithm for almost-sure reachability for BCSGs}

In this section the focus is on the qualitative almost-sure
reachability problem, i.e., starting with an object of type $T_i$,
decide whether the reachability value is 1 {\em and} there exists an
optimal strategy to achieve this value for the player aiming to
maximize the reachability probability. That is, the algorithm
presented here computes a set $F$ of variables (types), 
such that for any $x_i \in F$, 
starting from one object of type $T_i$ 
there is a strategy $\tau$ for the player aiming to reach the target
type $T_{f^*}$, such that no matter what the other player does,
almost-surely an object of type $T_{f^*}$ will be reached.  
We of course also wish to compute
such a strategy if it exists.  Before presenting the algorithm, we
give some preliminary results based on the results in
\cite{ESY-icalp15-IC}.

Following the definitions introduced in (\cite{ESY-icalp15-IC}, Section 5), a
\textit{linear degenerate (LD)-PPS} is a PPS where every polynomial
$P_i(x)$ is linear, containing no constant term (i.e., $P_i(x) =
\sum_{j=1}^{n}p_{ij}x_{j}$) and where the coefficients $p_{ij}$ sum to 1. Hence,
a LD-PPS has for LFP ($q^*$) and GFP ($g^*$) the all-0 and the all-1
vectors, respectively. Furthermore, a PPS that does not contain a
linear degenerate bottom strongly-connected component (i.e., a component
in the dependency graph that is strongly connected and has no edges
going out of it), is called a \textit{linear degenerate
  free(LDF)-PPS}. In other words, a LDF-PPS is a PPS that satisfies
the conditions of 
Lemma \ref{lemma:5.1}(ii)  below. 
Given a minimax-PPS $x=P(x)$, a policy $\tau$ for the min player is
called LDF if the resulting PPS for all max player policies $\sigma$, 
namely $x =
P_{\sigma, \tau}(x)$, is a LDF-PPS. Having introduced this, now we can
reference some known results from \cite{ESY-icalp15-IC} and give a
concurrent version (Lemma \ref{lemma:LDFpolicy}) of one of the Lemmas
from \cite{ESY-icalp15-IC}.

\begin{lemma}
	[cf. \cite{ESY-icalp15-IC}, Lemma 5.1] For any PPS, $x = P(x)$, exactly one of the following two cases holds:
	\begin{enumerate}[label=(\roman*)]
		\item $x = P(x)$ contains a linear degenerate bottom strongly-connected component (BSCC), S, i.e., $x_S = P_S(x_S)$ is a LD-PPS, and $P_S(x_S) \equiv B_Sx_S$, for a stochastic matrix $B_S$.
		\item every variable $x_i$ either is, or depends (directly or indirectly) on, a variable $x_j$ where $P_j(x)$ has one of the following properties:
		\begin{enumerate}
			\item[1.] $P_j(x)$ has a term of degree 2 or more,
			\item[2.] $P_j(x)$ has a non-zero constant term, i.e., $P_j(\mathbf{0}) > 0$ or
			\item[3.] $P_j(\mathbf{1}) < 1$.
		\end{enumerate}
	\end{enumerate}
	\label{lemma:5.1}
\end{lemma}

\begin{lemma}
	[cf. \cite{ESY-icalp15-IC}, Lemma 5.2] If a PPS, $x = P(x)$, has either GFP $g^* < \mathbf{1}$, or LFP $q^* > \mathbf{0}$, then x = P(x) is a LDF-PPS.
	\label{lemma:5.2}
\end{lemma}

\begin{lemma}
	[cf. \cite{ESY-icalp15-IC}, Lemma 5.5] For any LDF-PPS, $x = P(x)$, and $y < \mathbf{1}$, if $P(y) \le y$ then $y \ge q^*$ and if $P(y) \ge y$, then $y \le q^*$. In particular, if $q^* < \mathbf{1}$, then $q^*$ is the only fixed-point $q$ of $x = P(x)$ with $q < \mathbf{1}$.
	\label{lemma:5.5}
\end{lemma}

\begin{lemma}
	[cf. \cite{ESY-icalp15-IC}, Lemma 9.1] For a minimax-PPS, $x = P(x)$, if the GFP $g^* < \mathbf{1}$, then:
 	\begin{enumerate}
 		\item there exists a (mixed) LDF policy $\tau$ for the min player such that $g_{*,\tau}^* < \mathbf{1}$.
 		\item for any LDF min player's policy $\tau'$, it holds that $g^* \le q_{*,\tau'}^*$.
 	\end{enumerate}
 	\label{lemma:LDFpolicy}
\end{lemma}
\begin{proof}
For the first point, recall that since $g^* < \mathbf{1}$, the algorithm from the previous section will return a mixed static strategy(policy) $\tau$ for the min player such that $g_{*, \tau}^* < \mathbf{1}$. Thus for all max's strategies $\sigma: g_{\sigma, \tau}^* \le \sup_{\pi \in \Psi_1} g_{\pi, \tau}^* = g_{*, \tau}^* < \mathbf{1}$. By Lemma \ref{lemma:5.2}, all PPSs, $x = P_{\sigma, \tau}(x)$, are LDF, which results in the policy $\tau$ being LDF as well. 

Showing the second claim, let us fix any LDF policy $\tau'$ for the
min player. Notice that $g^* = P(g^*) = \inf_{\pi} P_{*, \pi}(g^*) \le
P_{*, \tau'}(g^*)$. In the resulting maxPPS, there exist a strategy
$\sigma$ for the max player such that $g^* \le P_{\sigma, \tau'}(g^*)
= P_{*, \tau'}(g^*)$. For every variable $x_i$ with $g_i^* = \max
\{g_1^*, \dots , g_{d_i}^*\}$ in the maxPPS, the strategy itself
chooses the successor in the dependency graph that maximizes
$g_i^*$. Now using Lemma \ref{lemma:5.5} with LDF-PPS $x = P_{\sigma,
  \tau'}(x)$ and $y := g^* < \mathbf{1}$, it follows that $g^* \le
q_{\sigma, \tau'}^* \le \sup_{\pi \in \Psi_1} q_{\pi, \tau'}^* = q_{*,
  \tau'}^*$.
\end{proof}

We now present the algorithm.  First, as a preprocessing step, we apply
the algorithm of Figure \ref{fig:QualNonReach}, which identifies in
P-time all the variables $x_i$ where $g_i^* = 1$. We 
then remove these
variables from the system, substituting the value 1 in their place.
We then simplify and reduce the resulting SNF-form
minimax-PPS into a reduced form, with GFP $g^* <1$. 
Note that the resulting reduced SNF-form minimax-PPS may contain some variables
$x_j$ of form M, whose corresponding matrix $A_j(x)$ has some
entries that contain the value 1 
rather than a variable (because we substituted 1 for removed 
variables $x_j$, where $g^*_j = 1$). 
Note also that in the reduced SNF-form 
minimax-PPS each variable $x_i$ of form Q has
an associated quadratic equation $x_i = x_j x_k$,
because if one of the variables (say $x_k$) on the right-hand side was set to $1$
during preprocessing, the resulting equation ($x_i  = x_j$) would have been declared
to have form L in the reduced minimax-PPS.
We 
henceforth assume that the minimax-PPS is in SNF-form, with $g^* < 1$, and we let $X$ be
its set of (remaining) variables.  We apply now the algorithm of
Figure \ref{fig:Qual-AS-Reach} to the minimax-PPS with $g^* <1$, which
identifies the variables $x_i$ in the minimax-PPS (equivalently, the
types in the BCSG), from which we can almost-surely reach the target
type $T_{f^*}$ (i.e., $g_i^* = 0$ {\em and} there is a strategy
$\tau^*$ for the player minimizing non-reachability probability that
achieves this value, no matter what the other player does).

\begin{figure}[ht]
	\begin{enumerate}
		\item Initialize $S := \{x_i \in X \mid P_i(\mathbf{0}) > 0$, that is $P_i(x)$ has a constant term $\}$. \\
			  Let $\gamma_0^i := \Gamma_{min}^i$ for every variable $x_i \in X - S$.  Let $t:=1$.
		\item Repeat until no change has occurred to $S$:
		\begin{enumerate}
                \item if there is a variable $x_i \in X - S$ of form L
                  where $P_i(x)$ contains a variable already in $S$, then
                  add $x_i$ to $S$.
                \item if there is a variable $x_i \in X - S$ of form Q
                  where both variables in $P_i(x)$ are already in $S$,
                  then add $x_i$ to $S$.
                \item if there is a variable $x_i \in X - S$ of form M
                  and if for all $a_{min} \in \Gamma_{min}^i$,
                  there exists a $a_{max} \in
                  \Gamma_{max}^i$ such that $A_i(x)_{(a_{max},a_{min})} \in S
\cup \{1\}$,
                  then add $x_i$ to $S$.
		\end{enumerate}
              \item For each $x_i \in X-S$ of form M, let:\\
                $\gamma_t^i := \{a_{min} \in \gamma_{t-1}^i \;|\;
                \forall a_{max} \in \Gamma_{max}^i, \; 
A_i(x)_{(a_{max},a_{min})} \not\in S \cup \{1\}\}$.  \  (Note that $\gamma_t^i
                \subseteq \gamma_{t-1}^i$.)
		\item Let $F := \{x_i \in X-S \;|\; P_i(\mathbf{1}) < 1$, or $P_i(x)$ is of form Q $\}$
		\item Repeat until no change has occurred to $F$:
		\begin{enumerate}
                \item if there is a variable $x_i \in X-(S \cup F)$ of
                  form L where $P_i(x)$ contains a variable already in $F$,
                  then add $x_i$ to $F$.
                \item if there is a variable $x_i \in X-(S \cup F)$ of form M
                  such that for $\forall a_{max} \in \Gamma_{max}^i$,
                  there is a min player's action $a_{min} \in
                  \gamma_t^i$ such that
                  $A_i(x)_{(a_{max}, a_{min})} \in F$, then add $x_i$ to
                  $F$.
		\end{enumerate}
		\item If $X = S \cup F$, {\bf return} $F$, and halt.
		\item Else, let $S := X - F$,  $t := t+1$,  and go to step 2.
	\end{enumerate}
	\caption{P-time algorithm for computing almost-sure reachability 
types $\{x_i \;|\; \exists \tau \in \Psi_2 \; (g_{*, \tau}^*)_i = 0\}$ for a minimax-PPS (in SNF), associated with a given BCSG.}
	\label{fig:Qual-AS-Reach}
\end{figure}

\begin{theorem}
Given a BCSG with minimax-PPS, $x=P(x)$, such that 
the GFP $g^* < \mathbf{1}$, the algorithm in 
Figure \ref{fig:Qual-AS-Reach} terminates in polynomial time and
returns the following set of variables:\\ $\{x_i \in X \;|\; 
\exists \tau \in \Psi_2 \; (g_{*, \tau}^*)_i = 0\}$.
	\label{theorem:Qual-AS-Reach}
\end{theorem}
\begin{proof}
  First, let us provide some notation and terminology for analyzing
  the algorithm.  The integer $t \geq 1$ represents the number of
  iterations of the main loop of the algorithm, i.e., the number of
  executions of steps 2 through 7 (inclusive; note that some of these
  steps are themselves loops).  Let $S_t$ denote the set $S$ inside
  iteration $t$ of the algorithm and just before we reach step 3 of
  the algorithm (in other words, just after the loop in step 2 has
  finished).  Similarly, let $F_t$ denote the set $F$ just before step
  6 in iteration $t$ of the algorithm.  We also define a new set,
  $K_t$, which doesn't appear explicitly in the algorithm.  Let $K_t
  := X - (S_t \cup F_t)$, for every iteration $t \geq 1$.  The set
  $\gamma_t^i$ in the algorithm denotes a set of moves/actions of the
  min player at variable $x_i$ (i.e. type $T_i$).\footnote{We shall
    show that $\gamma_t^i$, for $t \geq 1$, is a set of actions such
    that if the minimizer's strategy only chooses a distribution on
    actions contained in $\gamma^i_t$, for each variable $x_i$, then
    starting at any variable $x_j \in X - S_t$, the play will always
    stay out of $S_t$.}

  We now start the proof of correctness for the algorithm.
  Clearly, the algorithm terminates, i.e., step 6 eventually gets
  executed.  This is because (due to step $7.$) each extra 
iteration of the main 
loop must add at least one variable to the set $S \subseteq X$, and variables
are never removed from the set $S$.   It also follows easily that
the algorithm runs in P-time, since the main loop 
executes for at most $|X|$ iterations, and during each such 
iteration, each nested
loop within it also executes at most $|X|$ iterations.
So, the proof of correctness requires us to show that 
when the algorithm halts, the set
  $F$ is indeed the winning set for the minimizer (i.e., the player
  that aims to minimize the non-reachability probability). That is,
  we need to show that for all $x_i \in F$ 
there exists a (not-necessarily static) strategy $\tau$
for the minimizing player such that $(g^*_{*,\tau})_i = 0$, i.e.,
regardless of what strategy $\sigma$ the maximizer plays again $\tau$
the probability of {\em not} reaching
the target is $0$.     
On the other hand, if $x_i \in S$, 
we need to show that 
there is no such strategy $\tau$ for the minimizer
that forces $(g^*_{*,\tau})_i = 0$.  In fact, 
we will show that for all $x_i \in S$
the following stronger property $(**)_i$ holds:

\begin{itemize}
\item[$(**)_i$:]  There is a strategy $\sigma$ for the maximizing
player, such that for any strategy $\tau$ of the minimizing
player $(g^*_{\sigma,\tau})_i > 0$;  in other words, starting
with one object of type $T_i$, using strategy pair $\sigma$ and $\tau$,
there is a positive probability of never reaching the target type.
\end{itemize}
Note that property $(**)_i$ does not rule out that $g^*_i = 0$,
because even if $(**)_i$ holds it is
possible that $\inf_{\tau \in \Psi_2} (g^*_{\sigma,\tau})_i = 0$.

  First, let us show that if variable $x_i \in S$
when the algorithm terminates, then $(**)_i$ holds.
To show this,
  we use induction on the ``time'' when a variables is added to
  $S$. That is, if all variables $x_j$ added to $S$ in previous steps
  and previous iterations satisfy $(**)_j$, then if a new variable $x_i$ is
  added to $S$, it must also satisfy $(**)_i$. 
In the process of proving this, we shall in fact
  construct a single non-static randomized strategy $\sigma$ for the max
  player that ensures that for all $x_i \in S$, 
  regardless what strategy $\tau$ the min player plays 
  against $\sigma$, the probability of not reaching the target
  starting at one object of type $T_i$ 
  is positive.

Consider the initial set $S$ of variables $\{x_i \in X \mid 
P_i({\mathbf{0}}) > 0 \}$ that $S$ is initialized to in
  Step (1.) of the algorithm.  
Clearly all these variables satisfy $g_i^* \ge P_i(\mathbf{0})
  > 0$.  Thus, for these variables assertion $(**)_i$ holds using
{\em any} strategy $\sigma$ for the maximizer.
 Next consider a variable $x_i$ added to $S$ inside the loop in step (2.)
  of the algorithm, during some iteration.

\begin{enumerate}[label=(\roman*)]
\item If $x_i = P_i(x)$ is of form L, then $P_i(x)$ contains a
  variable $x_j$ (with a positive coefficient), that was added
  previously to $S$, and hence $(**)_j$ holds.  Thus there is a
  positive probability that one object of type $T_i$ will produce one
  object of type $T_j$ in the next generation.  It thus follows that
  $(**)_i$ holds, by using the same strategy $\sigma \in \Psi_1$ that
  witnesses the fact that $(**)_j$ holds.

\item If $x_i = P_i(x)$ is of form Q (i.e., $x_i = x_j \cdot x_r$),
  then $P_i(x)$ has both variables already added to $S$, i.e., $(**)_j$
  and $(**)_r$ both hold.  Then $(**)_i$ also holds, because starting
  from any object of type $T_i$, the next generation necessarily
  contains one object of type $T_j$ and one object of type $T_r$, and
  thus by combining the two witness strategies for $(**)_j$ and
  $(**)_r$, we have a strategy $\sigma \in \Psi_1$ that, starting from
  one object of type $T_i$, will ensure positive probability of not
  reaching the target, regardless of the strategy $\tau \in \Psi_2$ of
  the minimizer.

\item If $x_i = P_i(x)$ is of form M, then $\forall a_{min} \in
  \Gamma_{min}^i, \; \exists a_{max} \in \Gamma_{max}^i$ such that
  $A_i(x)_{(a_{max}, a_{min})} \in S \cup \{1 \}$.  
  In this case, let us define the
  strategy $\sigma$ to behave as follows at any object of type $T_i$:
  for each $a_{min} \in \Gamma_{min}^i$, we designate one ``witness''
  $a_{max}[a_{min}] \in \Gamma_{max}^i$, which witnesses that
  $A_i(x)_{(a_{max}[a_{min}], a_{min})} \in S \cup \{1\}$.  
  Then, at any object of type
  $T_i$, $\sigma$ chooses uniformly at random among the witnesses
  $a_{max}[a_{min}]$ for all $a_{min} \in \Gamma_{min}^i$.  So,
  starting with one object of type $T_i$, no matter what strategy the
  min player chooses, there is a positive probability that in the next
  step that object will either not produce any offspring (in the case 
where $A_i(x)_{(a_{max}[a_{min}], a_{min})} = 1$) and hence not reach the target, or else will
generate a single successor object of type
  $T_{(a_{max}[min],a_{min})}$, associated with variable
  $x_{(a_{max}[min],a_{min})}$ that already belongs to $S$, and hence such that
  $(**)_j$ holds. Hence, by combining with the strategies that witness
  such $(**)_j$ with the local (static) behavior of $\sigma$ described
  for any object of type $T_i$, we obtain a strategy $\sigma$ that
  witnesses the fact that $(**)_i$ holds.
\end{enumerate}

Now consider any variable $x_i$ that is added to $S$ in step (7.) of some
iteration $t$, in other words any variable $x_i \in K_t$. Since all
variables in $K_t$ were not added to $S_{t}$ or $F_t$ during iteration
$t$, we must have that: (A.) $x_i$ satisfies $P_i(\mathbf{1}) = 1$ and
$P_i(\mathbf{0}) = 0$; (B.) $x_i$ is not of $Q$ type; (C.) if $x_i$ is
of form L, then it depends directly only on variables in $K_t$; and
(D.) if $x_i$ is of form M, then

\begin{equation} 
\exists a_{max} \in
\Gamma_{max}^i \ \mbox{such that} \ \forall a_{min} \in \gamma_t^i, \;
A_i(x)_{(a_{max}, a_{min})} \not\in (F_t \cup S_t \cup \{ 1 \}) .
\label{formula-amax}
\end{equation}

Let $(q_h)_{h=0}^{\infty}, \; h \in \mathbb{N}$ be the infinite sequence of 
increasing probabilities defined by: $q_h = 2^{-(1/2^h)}$. Note that
as $h \rightarrow \infty$, the probability $q_h$ approaches $1$ from below.

Given a finite history $H$ of height $h$ (meaning the depth of the
forest that the history represents is $h$), for any object $e$ in the
current generation (the leaves) of $H$, if the object $e$ has type
$T_i$ such that the associated variable $x_i \in K_t$, we shall
construct the strategy $\sigma$ to behave as follows starting at the
object $e$.  The strategy $\sigma$ will choose one action $a_{max}$
that ``witnesses'' the statement (\ref{formula-amax}) above, and will
place probability
$q_h$ on that action, and it will distribute the remaining probability 
$1 - q_h$
uniformly among all actions in $\Gamma_{max}^i$.   We claim that
this strategy $\sigma$ ensures that for any object $e$ of type $T_i$
such that $x_i \in K_t$,  irrespective
of the strategy of the minimizing player, the probability of not reaching
the target type $T_{f^*}$ starting with $e$ (at any point in history)
is positive.   This clearly implies that 
$(g_{\sigma, *}^*)_{K_t} > \mathbf{0}$.   To prove this, there are two cases here:
\begin{enumerate}
\item First, suppose that during the entire play of the game,
at all objects $e$ whose type $T_i$ such that $x_i \in K_t$
has form $M$,
the min player only uses actions belonging to $\gamma_t^i$.
Then  in the resulting history
of play there {\em can not} be any such object $e$ whose 
child in the history 
(a necessarily unique child, since $e$ has type M)
is an object $e'$ of a type in $S_t$ (this is because step (3.)
of the algorithm, which defines $\gamma_t^i$, ensures that 
actions for the min player in $\gamma_t^i$ can not possibly produce
a child in $S_t$, no matter what the max player does).
Furthermore,  such an object $e$, occurring at 
depth $h$ in history, must with positive probability $\geq q_h$, 
produce a child $e'$ with a type in $K_t$ 
(because of point (D.) above, and because of the fact that the max player
 plays at $e$ a witness $a_{max}$ to the statement (\ref{formula-amax})
with probability $\geq q_h$).  

So consider an object $e$ of some type in $K_t$,
that occurs in a history $H$ at height $h \geq 0$,
and consider the tree of descendants of $e$.
What is the probability, under the strategy $\sigma$,
and under any strategy $\tau$ for the min player whose moves
are confined to the sets specified by $\gamma_t$, 
that the ``tree'' of descendants of $e$ 
is just a ``line'' consisting of   
an infinite sequence of objects $e_{0} = e$, $e_1$, $e_2$, $\ldots$,
all of which have types contained in $K_t$?
This probability is clearly  

$$\prod_{d=h}^{\infty}q_d = \prod_{d=h}^{\infty}2^{-(1/2^d)}  \geq
\prod_{d=0}^{\infty} 2^{-(1/2^d)} = 
 2^{- \sum^{\infty}_{d=0}  (1/2^d)}  =  2^{-2} =   \frac{1}{4}$$

 That is, irrespective of what strategy $\tau$ is played by the
minimizer,
there is 
positive probability bounded away from $0$ (indeed, $\geq 1/4$)
of staying forever confined in objects having types in
$K_t$.     In such a case, clearly, there will be positive
probability of not reaching the target type (since the types
in $K_t$ are not the target type).

\item Next suppose that, on the other hand,
there is a history $H$ of some height $h$ and
a leaf $e$ of $H$ that has type $T_i$ where $x_i \in K_t$,
such that the min player's strategy $\tau$ plays
at object $e$ some action(s) outside of the set
  $\gamma_t^i$ with positive probability. 
Note that for all actions $a'_{min} \not\in \gamma_t^i$, there is a max
player's action $a_{max} \in \Gamma_{max}^i$ such that 
$A_i(x)_{(a_{max}, a'_{min})}  \in S_t \cup \{1\}$.
Note moreover that the strategy $\sigma$ 
assigns positive probability, at least $(1-q_h)/|\Gamma^i_{max}|$
to every action in $\Gamma^i_{\max}$. 
Thus, %
if the min player's strategy $\tau$ puts positive
probability $\tau(H,e,a_{min}) > 0$ on some action $a_{min} \not\in \gamma^t_i$,
then with probability $\ge \big(\max_{a_{min}
    \not\in \gamma_t^i}{\tau(H,e, a_{min})} \big) \cdot 
\frac{(1 - q_h)}{|\Gamma_{max}^i|}$, 
either the object $e$ will have no child (since we can have 
$A_i(x)_{(a_{max},a_{min})} = 1$), or
the only child of object $e$ in the history
will be an object $e'$ whose type is in the set $S_t$, 
from which we already know
  that the target type $T_{f^*}$ is \textit{not} reached with positive
  probability.  So in either case, with positive probability 
the target type $T^*_{f^*}$
will not be reached from descendants of $e$.
\end{enumerate}

Now, let us assume the max player uses this strategy $\sigma$,
and suppose we start play at one object $e'$ of type $T_i$ 
such that $x_i \in K_t$.
Suppose, first, that during the entire history of play
the min player's strategy $\tau$ 
uses only  actions in $\gamma_t^i$ for all
variables $x_i \in K_t$ of form M.    In this case, 
with positive
probability bounded away from 0 (in fact $\geq 1/4$), the play tree
after $k$ rounds (i.e.,  depth $k$), for any positive $k \geq 1$, 
consists of simply a linear sequence of objects having types in $K_t$.
Thus in this case, with probability $\geq 1/4$, the play
will forever stay in $K_t$, and will
never reach target type $T_{f^*}$.   On the
other hand, suppose the min player's strategy $\tau$ does 
at some point
in some history consisting entirely
of a linear sequence of objects of types in $K_t$, 
namely at some specific object $e$ of type $K_t$ at depth $h$,
plays an action 
outside of
$\gamma_t^i$ with positive probability.  Then $\sigma$ ensures 
that with 
positive probability
(albeit a probability
depending on $h$ and thus not bounded away from $0$)
 either
$e$ will have no child or  
the unique child of $e$ 
will be an object of type $T_j$ such that $x_j \in S_t$, i.e., there is a 
positive probability 
of \textit{not} reaching the target
$T_{f^*}$ from the descendants of $e$, and thus also from the start
of the game (because we assumed the play staring from $e'$ and
up to $e$ consists of a linear sequence of objects all having
types in $K_t$).   Thus, 
for all strategies $\tau \in \Psi_2$,
and all $x_i \in K_t$, $(g_{\sigma, \tau}^*)_{i} > 0$.
Note however, that in general it may be the case that
$\inf_{\tau} (g^*_{\sigma,\tau})_i = 0$,  because in the
case when $\tau$ does play outside of $\gamma^i_t$,
the probability of not hitting the target type
is not bounded away from $0$  (it depends both on the depth
$h$ at which $\tau$ first moves outside of $\gamma^i_t$ with
positive probability, and it also depends on the probability
of that move, and for both reasons it can be arbitrarily close to $0$).
This establishes the first part of the proof, i.e., that
for every $x_i \in S$ the property $(**)_i$ holds.

\vspace*{0.1in}

Now we proceed to the second part of the proof.
Suppose $F$ is the set of variables output by the algorithm when it halts
(and that therefore $S = X - F$). 
Suppose the algorithm executed exactly $t^*$ iterations of the main loop
before halting (so that the value of $t$ just before halting is $t^*$).
 We will show that
there is a (randomized non-static) strategy $\tau$ of the minimizing player
such that, for all $x_i \in F$,   regardless what 
strategy $\sigma$ the maximizer employs, 
starting with on object of type $T_i$, the probability
of {\em not} reaching the target type is $0$.
In other words, that  $(g^*_{*,\tau})_i = 0$, which is what we want to prove.

Before describing $\tau$, we first describe a static 
randomized strategy (i.e., a policy) $\tau^*$ for the minimizing player, 
that will
eventually lead us toward a definition of $\tau$.

Specifically, we define the policy (randomized static strategy) $\tau^*$ 
as follows.
Let $\tau'$ be any LDF policy such that $g^*_{*,\tau'} < {\mathbf{1}}$.
Such an LDF policy $\tau'$ must exist, by Lemma 
\ref{lemma:LDFpolicy}(1.).  
For all variables $x_i \in S$, 
let 
$\tau^*(x_i)  := \tau'(x_i)$.
In other words, at all variables $x_i \in S$, 
let $\tau^*$ behave according to the exact same distribution on actions
as the LDF policy $\tau'$.
For every variable $x_i \in F$ of form M,  define $\tau^*$
as follows: note that $x_i$ must have entered $F$ in some
iteration of the inner loop in step (5.)(b) of the algorithm,
during the final iteration $t^*$ of the main loop.
Therefore, for all $a_{max} \in \Gamma^i_{max}$, there exists
a ``witness'' 
action $a_{min}[a_{max}] \in \gamma^i_{t^*}$ such that the associated 
variable $A_i(x)_{(a_{max},a_{min}[a_{max}])}$ was already in $F$, before $x_i$
was added to $F$.   For $x_i \in F$ we define the policy
$\tau^*$ at variable $x_i$, i.e., the distribution $\tau^*(x_i)$,
to be the uniform distribution over the set 
$\{ a_{min}[a_{max}] \in \gamma^i_{t^*} \mid  a_{max} \in \Gamma^i_{max} \}$
of such ``witnesses''.

We now wish to show that $\tau^*$, as defined, is itself an LDF 
policy.
Consider any fixed policy (i.e., static randomized strategy)
$\sigma$ for the max player, and consider the resulting system of
polynomial equations $x = P_{\sigma,\tau^*}(x)$.  For every variable
$x_i \in F$, 
consider the variables $x_i$ depends on
directly in the equation $x_i = (P_{\sigma,\tau^*}(x))_i$. 
Let's consider separately the cases, based on the form of 
equation $x_i = P_i(x)$:
(1) if $x_i = P_i(x)$ is of form L, then in $x_i = (P_{\sigma,\tau^*}(x))_i$
the variable $x_i$ depends
directly only on variables in $F$,
because otherwise it would have been added to set $S$;
(2) if $x_i$ is of form M,  then again it depends directly
only on variables
in $F$, because $\tau^*(x_i)$ only
puts positive probability on
actions in
$\gamma^i_{t^*}$; 
(3) if $x_i$ is of form Q, 
then $x_i$ depends directly on at least one variable in $F$, because
otherwise it would have been added to $S$. This implies that, in the
dependency graph of $x = P_{\sigma, \tau^*}(x)$, every variable in $F$
satisfies one of the three conditions in Lemma \ref{lemma:5.1}(ii)
(namely, 1. or 3.). So for every variable $x_i \in X$, 
consider paths in the dependency
graph of $x = P_{\sigma, \tau^*}(x)$ starting at $x_i$:

\begin{itemize}[topsep=0pt, parsep=0pt]
\item either there exists a path from $x_i$ in
this dependency graph to variable $x_j \in F$,
which in turn must have a path to a variable $x_{j'}$ 
such that either $P_{j'}({\mathbf{1}}) < 1$, or $x_{j'}$
has form $Q$.  In either case, this means that 
$x_i$ satisfies one of the conditions 
of Lemma \ref{lemma:5.1}(ii)  
(namely, either condition
(1.) or condition (3.));  Or
\item all paths from $x_i$ only contain variables in $S$. But 
for all variables $x_k \in S$,
$\tau^*(x_k)$ is  
exactly the same distribution as $\tau'(x_k)$,  and since
the LDF policy $\tau'$ was chosen so that
$g^*_{*,\tau'} < {\mathbf{1}}$, this 
means that there is a path from $x_i$ to a variable $x_j$ 
satisfying one
  of the three conditions in Lemma \ref{lemma:5.1}(ii)  (specifically,
condition (3.)).
\end{itemize}

Therefore, $x = P_{\sigma, \tau^*}(x)$ is a LDF-PPS. But since the fixed
strategy $\sigma$ was arbitrary, this implies that $\tau^*$ is indeed 
an LDF policy.
Since $\tau^*$ is LDF,
by Lemma \ref{lemma:LDFpolicy}(2.), it holds that $g^* \le q_{*,
  \tau^*}^*$. 
We now construct a {\em non-static} strategy $\tau$, which combines
the behavior of the 
two policies (i.e., two static strategies) $\tau'$ and $\tau^*$
in a suitable way, 
such that for all $x_i \in F$,
$(g^*_{*,\tau})_i = 0$.    In other words, $\tau$ will be a strategy
for the minimizer such that, no matter what strategy $\sigma$
the maximizer uses starting with one object
of type $T_i$, the probability of not reaching the target type is $0$.

The non-static strategy $\tau$ is defined as follows.
The strategy $\tau$ will, in each generation,
declare one object in the current generation to be
the ``queen'' (and this object will always have
a type in $F$).  Other objects in each generation
will be ``workers''.
Assume play starts
at a single object $e$ of some type $T_i$ such that $x_i \in F$.
We declare this object the ``queen'' in the initial population.
If the queen $e$ has associated variable $x_i$ of form $M$, then $\tau$
plays at $e$ according to distribution $\tau^*(x_i)$.   This results,
(with probability 1), regardless
of the strategy of the maximizer, in some successor object $e'$ 
in the next generation of type $T_j$
such that $x_j \in F$. In this case,
we declare $e'$ the queen in the next generation, 
and we apply the same strategy $\tau$ starting at the
queen $e'$ of the next generation, as if the game
is starting at this single object $e'$ of type $T_j$.  
If the variable $x_i$ associated with the queen $e$ is of form L, then in the next generation
either we hit the target (with probability $(1-P_i({\mathbf{1}}))$,
or (with probability $P_i({\mathbf{1}})$) 
we generate a single successor object $e'$ of some type $T_j$ such that $x_j \in F$.  
In this latter case again,
we declare  $e'$ the queen of the next generation, and 
we use the same strategy $\tau$ that is being defined, and
apply it to $e'$ as if the game is starting with the single
object $e'$.
If the queen $e$ has associated variable $x_i$ of form $Q$, 
then in the next generation there
are two successor objects, $e'$ and $e''$ of types 
$T_j$ and $T_k$ respectively (these may be the same type), 
such that either $x_j \in F$ or $x_k \in F$, or both are in $F$.
In this case, we choose one of the two successors whose 
type is in $F$, say wlog that this is $e'$, and
we declare $e'$ the queen of the next generation, 
we proceed from $e'$ using the same strategy $\tau$ that is being defined,
as if the game starts with the single object $e'$.
However, we declare the other object $e''$
a ``worker'', and starting with $e''$ and thereafter
(in the entire subtree of play rooted at $e''$) we use the static 
strategy (i.e., the LDF policy)  $\tau'$. 
This completes the definition of the non-static strategy $\tau$.

We now show that indeed $\tau$ satisfies that, no matter what strategy
$\sigma$ the maximizer uses against it, for any $x_i \in F$, starting
with one object of type $T_i$, the probability of not reaching the
target type is $0$.  In other words, we show that using $\tau$ the
probability of reaching the target type is $1$, no matter what
the opponent does.

To see this, first note that the LDF policy $\tau'$ was chosen
so that $g^*_{*,\tau'} < 1$.     Thus, 
since in the resulting max-PPS  $x=P_{*,\tau'}(x)$ the player
maximizing non-reachability probability always has a static
optimal strategy  (by Theorem \ref{theorem:GFP-NonReach}), it follows that
the subtree of the play rooted at any ``worker''
object $e''$ starting at which strategy $\tau'$ is applied by 
the min player, has positive probability $(1-g^*_{*,\tau'})_i > 0$
of eventually reaching the target type. 

Next note that the sequence of queens is finite if and only if we have
hit the target.
Next, we establish that if the sequence of queens
is infinite, then, with probability 1, infinitely often the queen
is of type Q and thus in the next generation it generates
both a queen and a worker.
Thus, because of the
infinite sequence of workers generated by queens,
there will be infinitely many independent chances
of hitting the target with probability at least 
$\min_i (1-g^*_{*,\tau'})_i$.
Hence,  we will hit the target (somewhere
in the entire tree of play) with probability 1.

It remains to show that, if the sequence of queens is infinite,
then, with probability 1, infinitely often a queen is of type Q. 
We in fact claim that with positive probability bounded
away from $0$, in the next $n = |X|$ generations 
either we reach a queen of type Q, or 
the queen has the target as a child.
To see this, we note that each type $x_i \in F$ has
entered $F$ in some iteration of the loop in step (5.) of the
algorithm (in the last iteration of the main loop).    
We can thus define inductively, for each
variable $x_i \in F$,  a finite tree $R_i$, rooted at $x_i$,
which shows ``why'' $x_i$ was added to $F$.
Specifically, if $P_i({\mathbf{1}}) < 1$ or $x_i$ has form Q,
then $R_i$ consists of just a single node (leaf) labeled by $x_i$.
If $x_i$ has form L, then it was added in step (5.) because 
$P_i(x)$ has a variable $x_j$ that was already in $F$.  In this case,
the tree $R_i$ has an edge from the root, labeled by $x_i$
to a single child labeled by $x_j$, such that this child
is the root of a subtree $R_j$.  If $x_i$ has form $M$
then $R_i$ has a root labeled by $x_i$ 
and has  children labeled by all variables $x_{(a_{max}, a_{min}[a_{max}])}
\in F$, and have $R_{((a_{max}, a_{min}[a_{max}])}$ as a subtree,
where $a_{max} \in \Gamma^i_{max}$ and where $a_{min}[a_{max}] \in 
\gamma^i_{t^*}$ is the ``witness'' for $a_{max}$, in the 
condition that allows step 5.(b) of the algorithm to add $x_i$ to $F$.

Clearly the tree $R_i$ is finite and has depth at most $n$ (since
there are only $n$ variables, and there is a strict order in
which the variables entered the set $F$).

Now we argue that starting at a queen of type $T_i$, using strategy $\tau$
for the minimizing player,  
with positive probability bounded away from $0$
in the next $n$ steps the sequence of queens will 
follow a root-to-leaf path in $R_i$, regardless of the strategy
of the max player.  To see this, note that if a node
is labeled by $x_j$  is of form L,
then the play will in the next step, with probability associated
with the transition in the BCSG move to the unique child (the
new queen) $x_{j'}$ that is the immediate child of the root in $R_j$,
and thus next will be at the root of the subtree $R_{j'}$.   If the node
is labeled by $x_j$ of form $M$, then irrespective of the 
distribution on actions
played by the max player, in the next step with positive probability
bounded away from $0$, we will move to a child $x_{a_{max},a_{min}[a_{max}]} \in F$
which is a child of the root in $R_j$,  
itself rooted at a subtree $R_{(a_{max},a_{min}[a_{max}])}$, 
because at queen objects we are using $\tau^*$ for the minimizer.
Thus, starting at a queen $x_i$,
with positive probability bounded away from $0$, within
$n$ steps the play arrives a leaf of the tree $R_i$.
If the leaf corresponds to a variable $x_j$ with $P_j({\mathbf{1}}) < 1$,
then the process will reach in the next step the target type with positive probability bounded away from $0$.
If, on the other hand, the leaf corresponds to a variable $x_j$ 
of form Q,  then the queen generates two children.  
The probability that the queen reaches infinitely often a leaf
of type L with $P_j({\mathbf{1}}) < 1$ but does not reach the target is 0.
Thus, if the queen never reaches the target throughout
the play, then the queen will generate more than one child
infinitely often with probability 1, and hence
will generate infinitely many independent workers with probability 1.
By the choice of the policy $\tau'$ followed by 
workers, the subtree rooted at each worker will hit the
target with positive probability bounded away from $0$.
Hence, the probability of hitting the target type is $1$. 
This completes the proof of the theorem.
\end{proof}

\vspace*{0.2in}

\begin{corollary}
Let $F$ be the set of variables output by the algorithm
in Figure \ref{fig:Qual-AS-Reach}.
\begin{enumerate}
\item Let $S = X-F$.
  There is a randomized 
non-static strategy $\sigma$ for the max
  player (maximizing non-reachability) such that for all $x_i \in S$,
and for all strategies $\tau$ of the min
player (minimizing non-reachability), starting with one object of type $T_i$, 
the probability of reaching the target type is $< 1$.
	\label{coroll:max-AS-OptPolicy}

\item There is a
  randomized non-static strategy $\tau$ for the min player (minimizing non-reachability), such that 
for all strategies $\sigma$ of the max player (maximizing non-reachability), 
and 
for all $x_i \in F$,
 starting at one object of type
  $T_i$ the probability of reaching the target
  type is $1$.
\label{corollary:min-AS-OptStrategy}
\end{enumerate}
\end{corollary}

\begin{proof} 

1. The strategy $\sigma$ 
constructed in the
proof of Theorem \ref{theorem:Qual-AS-Reach}
for 
variables $x_i \in S$
achieves precisely this.
2. The strategy $\tau$ constructed
in the proof of Theorem 
\ref{theorem:Qual-AS-Reach}
for all variables $x_i \in F$ achieves precisely
this.
\end{proof}

\noindent {\bf Remark:} Neither the strategy 
$\sigma$  from Corollary 
\ref{coroll:max-AS-OptPolicy}, nor the 
strategy 
$\tau$ from \ref{corollary:min-AS-OptStrategy},
both of which were constructed in the proof of Theorem 
\ref{theorem:Qual-AS-Reach}, are {\em static} strategies.
However, we note that both of these non-static randomized 
strategies have suitable compact descriptions (as functions
that map finite histories to distributions over actions
for objects in the current populations), and that
both these strategies can be constructed and described
compactly in polynomial time,
as a function of the encoding size of the input 
BCSG.\footnote{However, it is worth pointing out that
the functions that these 
strategies compute, i.e., functions from histories to distributions,
need not themselves be polynomial-time as a function
of the encoding size of the history: this is because the probabilities
on actions that are
involved can be double-exponentially small (and double-exponentially 
close to 1), as a function of the size of the history.}

\section{P-time algorithm for limit-sure reachability for BCSGs}

In this section, we focus on the qualitative limit-sure reachability
problem, i.e., starting with one object of a type $T_i$, decide whether
the reachability value is 1. Recall that there may not exist an
optimal strategy for the player aiming to reach the target $T_{f^*}$,
which was the question in the previous section (almost-sure
reachability). However, there may nevertheless be a sequence of
strategies that achieve values arbitrarily close to 1 (limit sure
reachability), and the question of the existence of such a sequence is
what we address in this section. Since we translate reachability into
non-reachability when analysing the corresponding minimax-PPS, 
we are asking whether
there exists a sequence of strategies $\langle \tau_{\epsilon_j}^*
\mid j \in \nat \rangle$ for the min player, such that $\forall j \in
\nat$, $\epsilon_{j} > \epsilon_{j+1} > 0$, and where $\lim_{j
  \rightarrow \infty} \epsilon_j =0$, such that the strategy
$\tau_{\epsilon_j}^*$ forces non-reachability probability to be at
most $\epsilon_j$, regardless of the strategy $\sigma$ used by the max
player. In other words, for a given starting object of type $T_i$, we
ask whether $\inf_{\tau \in \Psi_2} (g_{*, \tau}^*)_i = 0$.

\begin{figure}[!htb]
	\begin{enumerate}
        \item Initialize $S := \{x_i \in X \;|\; P_i(\mathbf{0}) > 0$,
          that is $P_i(x)$ has a constant term $\}$.
        \item Repeat until no change has occurred to $S$:
		\begin{enumerate}
                \item if there is a variable $x_i \in X-S$ of form L where
                  $P_i(x)$ contains a variable already in $S$, then add
                  $x_i$ to $S$.
                \item if there is a variable $x_i \in X-S$ of form Q where
                  both variables in $P_i(x)$ are already in $S$, then
                  add $x_i$ to $S$.
                \item if there is a variable $x_i \in X-S$ of form M and if
                  for all $a_{min} \in \Gamma_{min}^i$, there 
                  exists $a_{max} \in \Gamma_{max}^i$ such
                  that $A_i(x)_{(a_{max},
                  a_{min})} \in S \cup \{1\}$, then add
                  $x_i$ to $S$.
		\end{enumerate}
              \item Let $F := \{x_i \in X-S \;|\; P_i(\mathbf{1}) <
                1$, or $P_i(x)$ is of form Q $\}$
              \item Repeat until no change has occurred to $F$:
		\begin{enumerate}
                \item if there is a variable $x_i \in X- (S \cup F)$ of form L
                  where $P_i(x)$ contains a variable already in $F$, then
                  add $x_i$ to $F$.
                \item if there is a variable $x_i \in X- (S \cup F)$ of form M
                  and if the following procedure returns ``Yes", then
                  add $x_i$ to $F$.
                  \begin{enumerate}
                  \item Set $L_0 := \emptyset, \; B_0 := \emptyset, \;
                    k := 0$. Let $O := X - (S \cup F)$.
                  \item Repeat:
				\begin{itemize}
                                \item $k := k + 1$.
                                \item $L_k := \{a_{min} \in
                                  \Gamma_{min}^i -
                                  \bigcup_{j=0}^{k-1}L_j \;|\; \forall
                                  a_{max} \in \Gamma_{max}^i -
                                  B_{k-1}, \; A_i(x)_{(a_{max}, a_{min})}
                                  \in F \cup O \}$.
                                \item $B_k := B_{k-1} \cup \{a_{max}
                                  \in \Gamma_{max}^i - B_{k-1} \;|\;
                                  \exists a_{min} \in L_k \; s.t. \;
                                  A_i(x)_{(a_{max}, a_{min})} \in F\}$.
				\end{itemize}
				Until $B_k = B_{k-1}$.
                              \item Return: ``Yes'' if $B_k =
                                \Gamma_{max}^i$, and ``No'' otherwise.
                           \end{enumerate}
               \end{enumerate}
                          \item If $X = S \cup F$, {\bf return} $F$, and
                            halt.
		\item Else, let $S: = X - F$, and go to step 2.
	\end{enumerate}
	\caption{P-time algorithm for computing the  
types that satisfy limit-sure reachability in a given BCSG,
i.e., the set of variables $\{x_i \;|\; g_i^* = 0\}$ in the associated
minimax-PPS.}
	\label{fig:Qual-LS-Reach}
\end{figure}

Again, as in the almost-sure case,
we first, as a preprocessing step,
use the P-time algorithm from Proposition
\ref{prop:QualNonReach}
to remove all variables $x_i$ such that
$g_i^* = 1$, and we substitute $1$
for these variables in the remaining equations.  
We hence obtain a reduced SNF-form minimax-PPS, for which we 
can assume $g^* < 1$.
The set of all remaining variables in the
SNF-form minimax-PPS is again denoted by $X$.
Thereafter, we apply 
the algorithm in 
Figure 
\ref{fig:Qual-LS-Reach}, which 
computes the set of variables, $x_i$,
such that $g^*_i = 0$.  In other words,
we compute the set of types, such that  starting 
from one object of that type the
value of the reachability game is $1$.   
Before considering the
algorithm in Figure \ref{fig:Qual-LS-Reach} in detail, we 
provide some preliminary results that will be used
to prove its correctness. More precisely, we first examine the
nested loop in step 4.(b) of the algorithm.
This inner loop is derived directly from a closely
related  ``limit-escape'' construction used by
de Alfaro, Henzinger, and Kupferman in 
\cite{deAHK07}.   
For completeness, we provide proofs here for the facts we need
about this construction.

For a variable $x_i$ of form M, for 1-step local
strategies $\sigma(x_i)$ and $\tau(x_i)$ at $x_i$ for the two players
(i.e., $\sigma(x_i)$ and $\tau(x_i)$ are distributions on $\Gamma^i_{max}$ and $\Gamma^i_{min}$, respectively), and for a 
set $W \subseteq X \cup \{1 \}$
which can include both variables and possibly also the constant 1, 
let us define:
\begin{align*}
  p(x_i \rightarrow W, \sigma(x_i), \tau(x_i)) = 
\sum_{ \{ (a_{max}, a_{min}) \in \Gamma_{max}^i \times \Gamma_{min}^i 
\mid A_i(x)_{(a_{max},a_{min})} \in W\} }  \sigma(x_i)
  (a_{max}) \cdot \tau(x_i)(a_{min}) 
\end{align*}
Thus   $p(x_i \rightarrow W, \sigma(x_i), \tau(x_i))$ denotes
the probability that, starting with one object of type $T_i$,
and using the 1-step strategies specified by $\sigma(x_i)$ and $\tau(x_i)$,
we will either generate a child object of type $T_j$ such that $x_j \in W$,
or (only if $1 \in W$) generate no child object (i.e., go extinct in the
next generation).

Assume that in step 4.(b) for a variable $x_i$ the loop stops at
some iteration $m$ (i.e., $B_{m-1} = B_m$), but $B_m
\stackrel{\subset}{\neq} \Gamma_{max}^i$, and hence step 4.(b) answers
``No'', and $x_i$ is not added to $F$. In such a case, let us define the
following 1-step strategy, $\sigma(x_i)$ for the max player which will
be used in the next lemma. 
Let $D^i_{max} := \Gamma_{max}^i -
B_{m}$. Let
\begin{align}
    \sigma(x_i)(a_{max}) := 
\begin{cases}
    \cfrac{1}{|D^i_{max}|} & \quad \text{for every } a_{max} \in D^i_{max} \\[1em]
    0 & \quad \text{otherwise}
\end{cases}
	\label{eq:strategy_sigma}
\end{align}

\begin{lemma}
Suppose that for a variable $x_i \in X - (S \cup F)$ 
the answer in step 4.(b) of the algorithm is ``No'', and
  let $\sigma(x_i)$ be defined as in (\ref{eq:strategy_sigma}). Then,
  there is a constant $c_i > 0$ such that for every local 1-step strategy $\tau(x_i)$ for the min
  player at $x_i$, the following inequality holds:
	\begin{align*}
		p(x_i \rightarrow S \cup \{1\}, \sigma(x_i), \tau(x_i)) \ge c_i*p(x_i \rightarrow 
(F \cup S \cup \{1\}), \sigma(x_i), \tau(x_i))
	\end{align*}
	\label{lemma:Ratio-LS-Converge}
\end{lemma}
\begin{proof} Suppose
the loop from step 4.(b) stops at iteration $m$,
such that $B_{m-1} = B_m \subset \Gamma^i_{max}$.
There are two possibilities:
\begin{enumerate}
\item $L_m = \emptyset$: That is, 
for every $a_{min} \in \Gamma_{min}^i -
  \bigcup_{q=0}^{m-1}L_q$, there exists $a_{max} \in 
D^i_{max} = \Gamma_{max}^i -
  B_{m-1}$ such that $A_i(x)_{(a_{max}, a_{min})} \in S \cup \{1\}$. Let $\tau(x_i)$ be an
  arbitrary 1-step strategy for the min player and let $\sigma(x_i)$ 
be as defined in 
	\ref{eq:strategy_sigma}. Also let $D^i_{min} := \Gamma_{min}^i -
  \bigcup_{q=0}^{m-1}L_q$. Then it follows that:
	\begin{equation}
          p(x_i \rightarrow S \cup \{ 1 \}, \sigma(x_i), \tau(x_i)) \ge \sum_{a_{min} \in
            D^i_{min}} \frac{1}{|D^i_{max}|}\tau(x_i)(a_{min}) =
          \frac{1}{|D^i_{max}|} \sum_{a_{min} \in D^i_{min}} \tau(x_i)(
          a_{min})
\label{eq:low-S}
	\end{equation}
	Note that, by construction,
for all $a_{max} \in D^i_{max}$ and $a_{min} \in
       \bigcup_{q=0}^{m-1}L_q$, $A_i(x)_{(a_{max}, a_{min})} \in O$. 
Hence, since the support of distribution $\sigma(x_i)$ is 
$D^i_{max}$, and since $D^i_{min} = \Gamma_{min}^i -
  \bigcup_{q=0}^{m-1}L_q$, 
we have
\begin{equation} 
 p(x_i \rightarrow (F  \cup S \cup \{1\}), \sigma(x_i), \tau(x_i)) \le 
\sum_{a_{min} \in D^i_{min}}
\tau(x_i)(a_{min})
\label{eq:up-F}
\end{equation}

Combining these bounds, we get:
	\begin{eqnarray*}
		p(x_i \rightarrow S \cup \{1\}, \sigma(x_i), \tau(x_i)) & \ge &
\frac{1}{|D^i_{max}|} 
\sum_{a_{min} \in D^i_{min}} \tau(x_i)(a_{min})\\ & \ge & \frac{1}{|D^i_{max}|}p(x_i 
\rightarrow (F \cup S \cup \{1\}), \sigma(x_i), \tau(x_i))
	\end{eqnarray*}
      \item ${L_m \not= \emptyset}$, but 
$\{a_{max} \in D^i_{max} \;|\; \exists
        a_{min} \in L_m \; s.t. \; A_i(x)_{(a_{max}, a_{min})} \in F \} =
        \emptyset$. Therefore for all $a_{max} \in D^i_{max}$, and for all
        $a_{min} \in L_m$, $A_i(x)_{(a_{max}, a_{min})} \in O$. Let $\tau(x_i)$
        be any 1-step strategy for the min player, and 
let $\sigma(x_i)$ be as defined in \ref{eq:strategy_sigma}.
Let 
$D^i_{min} := \Gamma_{min}^i - \bigcup_{q=0}^{m}L_q$.
Note that if $D^i_{min} = \emptyset$, then 
	$p(x_i \rightarrow S \cup \{1\}, \sigma(x_i), \tau(x_i)) = 0
	= p(x_i \rightarrow (F \cup S \cup \{1\}), \sigma(x_i), \tau(x_i))$.
So, in this case, the lemma holds for any constant  $c > 0$.
If $D^i_{min} \neq \emptyset$, then 
both the inequalities (\ref{eq:low-S}) and 
(\ref{eq:up-F}) hold again, with the minor modification that now we have
$D^i_{min} = \Gamma_{min}^i - \bigcup_{q=0}^{m}L_q$ instead
of $D^i_{min} := \Gamma_{min}^i - \bigcup_{q=0}^{m-1} L_q$.
\end{enumerate}
Therefore, in both cases the lemma is satisfied with
$c_i := \frac{1}{|D^i_{max}|} = \frac{1}{|\Gamma_{max}^i - B_{m}|}$.
\end{proof}

We are now ready to prove correctness for the algorithm
in Figure \ref{fig:Qual-LS-Reach}.
\begin{theorem}
  Given a BCSG with minimax-PPS,  $x = P(x)$, with GFP $g^* < \mathbf{1}$,
the algorithm in Figure \ref{fig:Qual-LS-Reach}
  terminates in polynomial time, and returns the set of variables 
$\{x_i \in X \;|\; g_i^*  = 0\}$.
	\label{theorem:Qual-LS-Reach}
\end{theorem}
\begin{proof} 

The fact that the algorithm terminates and runs in polynomial time
is again evident, as in case of the almost-sure algorithm.
(The only new fact to note is that the new inner loop in
 step 4.(b), can iterate at most $\max_i |\Gamma^i_{max}|$ times because with
 each new iteration, $k$, at least one action is added to the $B_{k-1}$,
 or else the algorithm halts.)

We need to show that when 
the algorithm terminates, for all $x_i \in F$,  $g^*_i = 0$,
and for all $x_i \in S = X - F$, $g^*_i > 0$.

Let us first show that for all $x \in S$,  $g^*_i > 0$.
In fact, we will show that there is a
strategy $\sigma \in \Psi_1$, and a vector $b > 0$ of
values, such that
for all $x_i \in S$,  $(g^*_{\sigma,*})_i \geq b_i > 0$.
For the base case, since any variable $x_i$ contained in $S$ at
the initialization step has $g_i^* \ge P_i(\mathbf{0}) > 0$, 
we have $(g^*_{\sigma,*})_i > P_i(\mathbf{0}) > 0$ 
for {\em any} strategy $\sigma$, so let $b_i := P_i(\mathbf{0})$. 
For the inductive step, first consider any variable $x_i$
added to $S$ in step 2, in some iteration of the main loop
of the algorithm.
\begin{enumerate}[label=(\roman*)]
\item If $x_i = P_i(x)$ is of form L, then $P_i(x)$ has a variable
  $x_j$ already in $S$, 
and by induction $(g^*_{\sigma, *})_j \geq b_j > 0$. Since $P_i(x)$ is
  linear, with a term $q_{i,j} \cdot x_j$,
such that $q_{i,j} > 0$, 
we see that $(g^*_{\sigma,*})_i \ge q_{ij}  \cdot b_j > 0$,
so let $b_i :=  q_{ij}  \cdot b_j$.
\item If $x_i = P_i(x)$ is of form Q (i.e., $x_i = x_j \cdot x_r$),
  then $P_i(x)$ has both variables previously added to $S$,
  i.e., $(g^*_{\sigma,*})_j \ge b_j > 0$ and $(g*_{\sigma,*})_r \ge b_r > 0$. 
Then clearly $(g^*_{\sigma,*})_i \geq b_j \cdot b_r > 0$.
 So let $b_i := b_j
  * b_r$.
\item If $x_i = P_i(x)$ is of form M, then $\forall a_{min} \in
  \Gamma_{min}^i$, $\exists a_{max} \in \Gamma_{max}^i$ such that
  $A_i(x)_{(a_{max}, a_{min})} \in S \cup \{1\}$. 
  For each $a_{min} \in \Gamma_{min}^i$, let us use
  $a_{max}[a_{min}] \in \Gamma_{max}^i$, to denote a ``witness''
  to this fact, i.e., such that $A_i(x)_{(a_{max}[a_{min}],a_{min})} \in S \cup \{1 \}$.
 Let strategy $\sigma$ 
  do as follows:  in any object of type $T_i$ corresponding to $x_i$,
  $\sigma$ selects uniformly at random an action from 
   the set  $\{ a_{max}[a_{min}] \in \Gamma_{max}^i \mid  
a_{min} \in \Gamma_{min}^i \}$ of all such witnesses. 
  Clearly then, for any $a_{min} \in \Gamma^i_{min}$,
  the probability that $\sigma$ at an object of type $T_i$ will
  choose the witness action $a_{max}[a_{min}]$ is 
  at least $\frac{1}{|\Gamma_{max}^i|}$ (and in fact is also at least $\frac{1}{|\Gamma_{min}^i|}$). 
  So, using $\sigma$, starting with one object
of type $T_i$, no matter what strategy
           the min player chooses, there is a positive probability 
$\geq \frac{1}{|\Gamma_{max}^i|}$ that either the object will have no child or
           the object will generate a single 
child object
of type $T_{(a_{max},a_{min})}$, associated
           with variable $x_j = A_i(x)_{(a_{max},a_{min})} \in S$,
and hence such that $(g^*_{\sigma,*})_{j} \geq b_{j} > 0$.
So no matter what strategy
  the min player picks, there is at least $\frac{1}{|\Gamma_{max}^i|}$
  probability that the unique 
child object  belongs to $S$, or that there is no child object.
  Hence, $(g^*_{\sigma,*})_i \ge \frac{1}{|\Gamma_{max}^i|} *
  \min \{b_j \mid  x_j \in S\} >
  0$, and again we let $b_i := \frac{1}{|\Gamma_{max}^i|} * \min
  \{b_j \mid  x_j \in S \}$.
\end{enumerate}

Now consider any variable $x_i$ added to $S$ in step 6 at some
iteration of the algorithm (i.e., $x_i \in K := X - (S \cup
F)$). Because $x_i$ was not previously added to $S$ or $F$, then: (A.)
$x_i$ satisfies $P_i(\mathbf{0}) = 0$ and $P_i(\mathbf{1}) = 1$; (B.)
$x_i$ is not of type Q; (C.) if $x_i$ is of form L, then it depends
directly only on variables in $K$; and (D.) if $x_i$ is of type M, then
the answer for $x_i$
in step 4.(b) (during the latest iteration
of the main loop) was ``No''.

For each $x_i \in K$ of type $M$, let $\sigma(x_i)$ be 
a probability distribution on actions in $\Gamma_{max}^i$ defined in
(\ref{eq:strategy_sigma}).   Let strategy $\sigma$ use the 
local 1-step strategy $\sigma(x_i)$ at every object of type $T_i$
encountered during history.
We show that, for every $x_i \in K$,
$(g_{\sigma,*}^*)_i \ge b_i$ for some $b_i > 0$.

By Lemma
\ref{lemma:Ratio-LS-Converge}, for each variable $x_i
\in K$ of type M, and for 
any arbitrary 1-step strategy $\tau(x_i)$  for the min player
at $x_i$, there exists $c_i > 0$ such that:
\begin{align*}
	p(x_i \rightarrow S \cup \{1\}, \sigma(x_i), \tau(x_i)) \ge c_i * 
p(x_i \rightarrow 
(F \cup S \cup \{1\}), \sigma(x_i), \tau(x_i))
\end{align*}

For $r \geq 1$, let $Pr_{x_i}^{\sigma, \tau}(K \Until_{=r}
(S \cup \{1\}) )$ denote the probability that, starting with one object of
type $T_i$, where $x_i \in K$, using strategy $\sigma$ and 
an arbitrary (not necessarily static) strategy $\tau$, 
the history of play will stay in the set $K$  
for $r-1$ rounds, and in the $r$'th 
will either transition to
an object whose type is in the 
set $S$, or will die (i.e., produce no children).  Define $Pr_{x_i}^{\sigma, \tau}(K \Until_{=r}
(F \cup S \cup \{1\}))$ similarly. 
The following claim is a simple corollary
of 
Lemma \ref{lemma:Ratio-LS-Converge}. 
Let $c := \min \{c_i \mid x_i \in K \}$. (Note that
$0 < c \leq 1$.)   

\begin{claim} 
For any integer $r \geq 1$, and for 
any (not necessarily static) strategy $\tau$
for the min player, 
  $Pr_{x_i}^{\sigma, \tau}(K \Until_{=r} (S \cup \{1\})) \ge c * Pr_{x_i}^{\sigma,
   \tau}(K \Until_{=r} (F \cup S \cup \{1\}))$.
\label{claim:r-step-reach-S-F}
\end{claim}
\begin{proof}

Let $H(x_i,K, r-1)$ denote the set of all sequence of types in $K$ of length
$r-1$, starting with $x_i \in K$.  
For a history (sequence) $h \in H(x_i,K, r-1)$,
let $l(h)$ denote the index of the variable associated with the
last type in $h$, i.e., the one occurring at round $r-1$. 
For each $h \in H(x_i, K, r-1)$
there is some probability $q_h \geq 0$
that, starting at $x_i \in K$, the population 
follows the history $h$ for $r-1$ rounds.
So 

\begin{eqnarray*} 
Pr_{x_i}^{\sigma, \tau}(K\Until_{=r} (S \cup \{1\})) & = & \hspace*{-0.1in} \sum_{h \in H(x_i,K, r-1)} q_h  
\cdot
p(x_{l(h)} \rightarrow S \cup \{1\}, \sigma(h), \tau(h))\\
& \geq & \hspace*{-0.1in} \sum_{h \in H(x_i, K, r-1)}  q_h \cdot c_{l(h)} \cdot  
p(x_{l(h)} \rightarrow (F \cup S \cup \{1\}), \sigma(h), \tau(h))  \quad \
\mbox{(by Lemma \ref{lemma:Ratio-LS-Converge})}\\
& \geq & c \cdot \hspace*{-0.1in} \sum_{h \in H(x_i, K, r-1)} q_h  \cdot p(x_{l(h)} \rightarrow (F \cup S \cup \{1\}), 
\sigma(h), \tau(h))  \\ 
& = &   c \cdot Pr_{x_i}^{\sigma,
   \tau}(K \Until_{=r} (F \cup S \cup \{1\}))
\end{eqnarray*}
\end{proof}

We now argue that for all $x_i \in K$,
there exists $b_i > 0$ such that for any strategy $\tau$
for the min player,  $(g^*_{\sigma,\tau})_i > b_i > 0$.

Consider any strategy $\tau$ for the min player.
For $x_i \in K$,
let $Pr_{x_i}^{\sigma, \tau}(\Box K)$ denote the probability
that the history stays forever in $K$, starting at one
object of type $T_i$. 
Let $Pr_{x_i}^{\sigma, \tau}(K \Until
(S \cup \{1\}))$ denote the probability that
the history stays in set $K$  until it eventually either 
dies (has no children) or  transitions to
an object with type in
set $S$. Note that:

\begin{eqnarray*}
(g^*_{\sigma,\tau})_i  & \geq &
Pr_{x_i}^{\sigma, \tau}(\Box K)
+ Pr_{x_i}^{\sigma,\tau}(K \Until (S \cup \{1\})) \cdot 
\min \{(g_{\sigma,*}^*)_j \;|\; x_j \in S\}\\
& \geq & Pr_{x_i}^{\sigma, \tau}(\Box K)
+ Pr_{x_i}^{\sigma,\tau}(K \Until (S \cup \{1 \} )) \cdot \min \{b_j \mid x_j \in S\}
\end{eqnarray*}

We will show that, regardless
of the strategy $\tau$ for the min player, this probability must be at least:
$$b_i := \frac{c}{2}  \cdot \min \{b_j \mid x_j \in S\}$$
where $c := \min \{ c_i  \mid x_i \in K \}$.  Recall that $0 < c \leq 1$.
Let $p = Pr_{x_i}^{\sigma, \tau}(\Box K)$.   If $p \geq \frac{c}{2}$,
then we are done, since the inequalities above imply 
$(g^*_{\sigma,\tau})_i \geq \frac{c}{2}  \geq b_i$.
So, suppose $p < \frac{c}{2}$.
Observe that:
\begin{eqnarray*}
Pr_{x_i}^{\sigma,\tau}(K \Until (S \cup \{1\})) & = &
Pr_{x_i}^{\sigma,\tau}((K \Until (S \cup \{1\})) \cap \neg \Box K)\\
& = &  Pr_{x_i}^{\sigma,\tau}((K \Until (S \cup \{1\})) \mid \neg \Box K)
\cdot Pr_{x_i}^{\sigma,\tau}(\neg \Box K )\\
& = &  Pr_{x_i}^{\sigma,\tau}(K \Until (S \cup \{1\}) \mid  \neg \Box K)
\cdot (1-p) \\ 
& \geq &
Pr_{x_i}^{\sigma,\tau}(K \Until (S \cup \{1\}) \mid \neg \Box K) \cdot \frac{1}{2}.
\end{eqnarray*}

So it only remains to show that
$Pr_{x_i}^{\sigma,\tau}(K \Until (S \cup \{1\}) \mid \neg \Box K) \geq c$.
Note that the event $\neg \Box K$ is equivalent to the event
$(K \Until (F \cup S \cup \{1\}))$. 
The event $K \Until (S \cup \{1\})$ is equivalent to 
the disjoint union $\bigcup^\infty_{r = 1} K \Until_{=r} (S \cup \{1\})$.
Likewise for the event $K \Until (F \cup S \cup \{1\})$. 
Therefore:

\begin{eqnarray}
\nonumber Pr_{x_i}^{\sigma,\tau}(K \Until (S \cup \{1\}) \mid \neg \Box K) 
& = & 
\frac{Pr_{x_i}^{\sigma,\tau}(K \Until (S \cup \{1\}))}{Pr_{x_i}^{\sigma,\tau}(\neg \Box K)}\\
& = & \frac{\sum_{r=1}^{\infty}Pr_{x_i}^{\sigma,\tau}(K \Until_{=r} (S \cup \{1\}) )}
{\sum^\infty_{r=1}Pr_{x_i}^{\sigma,\tau}(K \Until_{=r} (F \cup S \cup \{1\}))}
\label{eq:lowbound-conditional}
\end{eqnarray}
But by Claim \ref{claim:r-step-reach-S-F},
for all $r \geq 1$,    $Pr_{x_i}^{\sigma,\tau}(K \Until_{=r} (S \cup \{1\})) \geq
c \cdot Pr_{x_i}^{\sigma,\tau}(K \Until_{=r} (F \cup S \cup \{1\}))$.
Hence, summing over all $r$, we have
$\sum_{r=1}^\infty Pr_{x_i}^{\sigma,\tau}(K \Until_{=r} (S \cup \{1\})) \geq
c \sum_{r=1}^\infty \cdot Pr_{x_i}^{\sigma,\tau}(K \Until_{=r} (F \cup S \cup \{1\}))$.
Hence,  dividing out and using (\ref{eq:lowbound-conditional}), 
we have $Pr_{x_i}^{\sigma,\tau}(K \Until (S \cup \{1\}) \mid \neg \Box K) \geq c$. 

Thus, $(g^*_{\sigma,\tau})_i \geq b_i$, and since this holds
for an arbitrary strategy $\tau$ for the min player, we have
$(g^*_{\sigma,*})_i \geq b_i > 0$.

\vspace*{0.1in}

We next want to show that if $F$ is the set of variables
output by the algorithm when it halts, then
for all
variables $x_i \in F$,  $g^*_i = 0$, or in other 
words, that the following holds:
\begin{align}
  \forall \epsilon > 0,\; \exists \tau_{\epsilon} \in \Psi_2
  \;\;s.t.\;\; \forall \sigma \in \Psi_1, \;\; (g_{\sigma,
    \tau_{\epsilon}}^*)_i \le \epsilon
	\label{eq:LS}
\end{align}

Let $N := \max_i |\Gamma_{min}^i|$. 
Given some $0 \le e \le
\frac{1}{2N}$, 
consider the following static distribution, $\mathit{safe}(x_i, e)$, on
actions for the min player at $x_i$  (i.e., distribution on $\Gamma^i_{min}$):
\begin{align}
    \mathit{safe}(x_i, e)(a_{min}) := 
\begin{cases}
  \big(e^2\big)^{j-1} \cdot \cfrac{ \big(1 - e^2\big) }{|L_j|} 
& \quad \text{if} \ a_{min} \in L_j , \  \text{for some} \  j \in \{1, \ldots, k-1\} \\[1em]
  \big(e^2\big)^{k-1} \cdot \cfrac{1}{|\Gamma_{min}^i -
    \bigcup_{q=0}^{k-1} L_q|} & \quad \text{otherwise}
\end{cases}
	\label{eq:strategy_tau_eps}
\end{align}

Given an $\epsilon >0$, we define a (static) strategy  $\tau_{\epsilon}$ as follows.
If a variable $x_i$ of form M is in $S$, then we let $\tau_{\epsilon} (x_i)$
be the uniform distribution on the corresponding action set $\Gamma^i_{min}$.
For variables in $F$, we define $\tau_{\epsilon}$ as follows.
Consider the last execution of the main loop of the algorithm.
Let $F_0 = \{x_i \in X-S \;|\; P_i(\mathbf{1}) <
                1$, or $P_i(x)$ is of form Q $\}$ be the set of variables assigned to $F$ in Step 3, and let
$x_{i_1}, x_{i_2}, \ldots, x_{i_{k^*}}$ be the variables
in $F -F_0$ ordered according to the time at which they were
added to $F$ in the iterations of Step 4. 
For each variable $x_{i_t} \in F$ of form M 
we let $\tau_{\epsilon} (x_{i_t}) = safe(x_{i_t} , e_t)$ where the
parameters $e_t$ are set as follows.
Let $n$ be the number of variables, and $N := \max_i |\Gamma^i_{min}|$ 
the maximum number
of actions of player min for any variable of form M.
Let ${\kappa}$ be the minimum of (1) $1/N$, (2) the minimum (nonnegative) coefficient
of a monomial in $P_i(x)$ over all variables $x_i$ of form L,
and (3) the minimum of $1-P_i(\mathbf{1})$ over all $x_i$ of form L
such that $P_i(\mathbf{1}) < 1$.
Let $\lambda= {\kappa}^n$. 
Clearly, $\lambda$ is a rational number that depends on the given minimax-PPS
$x=P(x)$ (and the corresponding BCSG) and it has polynomial
number of bits in the size of $P$.
Let $d_0 = \lceil \log (\frac{n}{ \epsilon \lambda}) \rceil$
and let $d_t = d_0 \cdot (2N)^t$ for $t \geq 1$.
We set $e_t = 2^{-d_t}$ for all $t \geq 0$.
The numbers $e_t$ can be doubly exponentially small, 
but they can be represented compactly in floating point,
i.e., in polynomial size in the size of $P$ and of $\epsilon$.
Note from the definitions that $e_0 \leq \epsilon \lambda /n$,
and $e_t = (e_{t-1})^{2N}$ for all $t \geq 1$.

Consider the max-PPS $x=P_{*,\tau_{\epsilon}}(x)$ obtained from the
given minimax-PPS $x=P(x)$ by fixing the strategy of the min player to
$\tau_{\epsilon}$.  For every variable $x_i$ of form L or Q, the
corresponding equation $x_i = P_i(x)$ stays the same, and for every
variable $x_i$ of form M the equation becomes $x_i = \max_{a_{max} \in
  \Gamma^i_{max}} \sum_{a_{min} \in \Gamma^i_{min}}
\tau_{\epsilon}(x_i)(a_{min}) \cdot A_i(x)_{(a_{max},a_{min})}$.  Let
$f^* = g^*_{*, \tau_{\epsilon}}$ be the greatest fixed point of the
max-PPS $x=P_{*,\tau_{\epsilon}}(x)$, and let $M = \max \{ f^*_i | x_i
\in F \}$.  We will show that $M \leq \epsilon$, i.e., $f^*_i \leq
\epsilon$ for all $x_i \in F$.

First, we show that all variables of $X$ have value strictly 
less than 1 in $f^*$, 
and we also bound the value of the variables of $S$ in terms of $M$.

\begin{claim} \label{claim:less-1} \mbox{}\\
(1) For all $x_i \in X$, $f^*_i < 1$.\\
(2) For all $x_i \in X$, $f^*_i \leq \lambda M +(1-\lambda )$.
\end{claim}
\begin{proof}
The algorithm of Proposition \ref{prop:QualNonReach} (see Fig. \ref{fig:QualNonReach}) computes the set $X$ of variables $x_i$ of the minimax-PPS such that $g^*_i <1$ (this set is denoted $S$ in Fig. \ref{fig:QualNonReach}, but to avoid confusion with the set $S$ of the limit-sure reachability algorithm of Fig. \ref{fig:Qual-LS-Reach}, we refer to it as $X$ in the following).
We use induction on the time that a variable $x_i$ was added to $X$ in
the algorithm of Fig. \ref{fig:QualNonReach} to show the claim. 
For part (2), our induction hypothesis is that
if a variable $x_i$ is added to $X$ at time $t$
(where the initialization is time 1) then 
$f^*_i \leq {\kappa}^t M +(1-{\kappa}^t)$.
This inequality implies (2) since $t \leq n$ and $\lambda={\kappa}^n$.

For the basis case ($t=1$), 
$x_i$ is a deficient variable, i.e. $P_i(\mathbf{1}) <1$, hence
$f^*_i \leq P_i(\mathbf{1}) \leq 1-{\kappa} < 1$.

For the induction step, if $x_i$ is of form L or Q, then
$P_i(x)$ contains a variable $x_j$ that was added earlier to $X$,
hence $f^*_i <1$ follows from $f^*_j <1$ by the induction hypothesis.
For part (2), if $x_i$ is of form L, then the coefficient
of $x_j$ in $P_i(x)$ is at least ${\kappa}$ and 
$f^*_j \leq {\kappa}^{t-1} M +(1-{\kappa}^{t-1})$ by the induction
hypothesis, hence 
$f^*_i \leq {\kappa}({\kappa}^{t-1} M +(1-{\kappa}^{t-1})) + 1-{\kappa}= {\kappa}^t M +(1-{\kappa}^t)$.
If $x_i$ is of form Q, then 
$f^*_i \leq f^*_j \leq {\kappa}^{t-1} M +(1-{\kappa}^{t-1}) \leq {\kappa}^t M +(1-{\kappa}^t)$. 

If $x_i$ is of form M then
for every action $a_{max} \in \Gamma^i_{max}$,
there exists an action $a_{min} \in \Gamma^i_{min}$ such that
the variable $x_j =A_i(x)_{(a_{max},a_{min})}$ was added previously to $X$,
and hence its value in $f^*$ is $<1$ by the induction hypothesis.
Since $\tau_{\epsilon}(x_i)$ plays all the actions of $\Gamma^i_{min}$
with nonzero probability, both when $x_i \in S$ and
when $x_i \in F$, it follows that $f^*_i <1$.
This shows part (1).
For part (2), if $x_i \in F$, then
$f^*_i \leq M \leq {\kappa}^t M +(1-{\kappa}^t)$, where the first inequality follows
from the definition of $M$.
Suppose $x_i \in S$ and let $a_{max}$ be an action in $\Gamma^i_{max}$
that yields the greatest fixed point $f^*_i$ in the max-PPS equation
$x_i =  (P_{*,\tau_{\epsilon}}(x))_{i}$. The right-hand side 
for this
action is a linear expression that contains a
variable $x_j =A_i(x)_{(a_{max},a_{min})}$ that was added previously to $X$, and the coefficient of this term is 
$1/|\Gamma^i_{min}| \geq 1/N \geq {\kappa}$, since $\tau_{\epsilon}(x_i)$
is the uniform distribution for $x_i \in S$.
Therefore, $f^*_i \leq {\kappa} f^*_j +(1-{\kappa}) \leq {\kappa}({\kappa}^{t-1} M +(1-{\kappa}^{t-1})) + 1-{\kappa}= {\kappa}^t M +(1-{\kappa}^t)$.
\end{proof}

We can show the key lemma now.

\begin{lemma}
For all $x_i \in F$, $f^*_i \leq \epsilon$.
\end{lemma}
\begin{proof}
Recall that $F= F_0 \cup \{x_{i_1}, x_{i_2} , \ldots, x_{i_{k^*}} \}$.
Let $M_0 = \max\{ f^*_i | x_i \in F_0 \}$
and let $M_t = f^*_{i_t}$ for $t \geq 1$ 
be the value of $x_{i_t}$ in
the greatest fixed point $f^*$ of the max-PPS 
$x= P_{*,\tau_{\epsilon}}(x)$.
Thus, $M= \max \{ M_t | t \geq 0 \}$.
Let $r_t = (e_t)^{2N-1}$. 
Note that for every $x_{i_t} \in F$ of form M,
the probability with which $\tau_{\epsilon}(x_{i_t})=safe(x_{i_t},e_t)$ plays 
any action in a set $L_j$ is at least 
$(e_t^2)^{N-1} (1-e_t^2)/N$ which is $> (e_t)^{2N-1} =r_t$
because $e_t < 1/(2N)$.
Let $s_t = \Pi_{j=1}^t r_j$; by convention, $s_0 =1$.

We will show first that for all $t \geq 0$,
there exist $a_t, g_t \geq 0$ that satisfy 
$a_t \geq \lambda \cdot s_t$ and $g_t \leq  t \cdot e_0 \cdot a_t/ \lambda$, and such that
$M_t \leq a_t M^2  +(1-a_t - g_t)M+ g_t$.
We will use induction on $t$.

Basis: $t=0$. Then $M_0 =f^*_i$ for a variable
$x_i \in F_0$ which is either a deficient variable of form L
or a variable of form Q. 
If $x_i$ is of form L, then note that
(1) $P_{i}$ does not contain a constant term 
(because otherwise $x_{i}$ would have been added to
set $S$ in Step 1), (2) all the variables of $P_{i}(x)$ are not in $S$ (because otherwise $x_{i}$ would have been added to
set $S$ in Step 2), hence they are all eventually added to $F$
and thus their value in $f^*$ is at most $M$,
and (3) the coefficients sum to at most $1-{\kappa}$ because
$P_i(\mathbf{1}) < 1$.
Therefore, $M_0 =f^*_i \leq (1-{\kappa})M \leq \lambda M^2 +(1- \lambda )M$.
If $x_i$ is of form Q, at least one of the variables of $P_i(x)$ must belong to $F$
(because otherwise $x_i$ would have been added to $S$ in Step 2), 
hence its value in $f^*$ is at most $M$, and the value of the other variable is
at most $\lambda M+(1- \lambda)$ by Claim \ref{claim:less-1}.
Therefore, $M_0 =f^*_i \leq M( \lambda M+1-\lambda ) = \lambda M^2 +(1-\lambda)M$.
Thus in both cases, $M_0 \leq  \lambda M^2 +(1-\lambda)M$.
We can take $a_0 =\lambda$, $ g_0 = 0$.

Induction step: We have $M_t = f^*_{i_t}$.
If $x_{i_t}$ is of form L, then $P_{i_t}(x)$ contains
a variable $x_j$ that was added earlier to $F$; its coefficient, say $p$,
is at least ${\kappa}$.
Note again that $P_{i_t}(x)$ does not contain a
constant term, all the other variables of $P_{i_t}(x)$ are not in 
$S$, hence they are all eventually added to $F$
and their value in $f^*$ is at most $M$, and the sum of their
coefficients is $1-p$.
Since the variable $x_j$ was added earlier to $F$,
by the induction hypothesis we have $f^*_j \leq a_u M^2  +(1-a_u - g_u)M+ g_u$
for some $u \leq t-1$.
Therefore, $M_t \leq p(a_u M^2  +(1-a_u - g_u)M+ g_u) + (1-p)M$
$= a_t M^2  +(1-a_t - g_t)M+ g_t$,
with $a_t = p a_u $ and $g_t = p g_u$.
Since $u \leq t-1$, we have $a_u \geq \lambda \cdot s_{u} \geq \lambda \cdot s_{t-1}$,
and since $p \geq {\kappa} \geq r_t$ it follows that 
$a_t =pa_u \geq \lambda \cdot s_{t-1} \cdot r_t = \lambda \cdot s_t$.
Also, $g_t = p g_u \leq p u e_0 a_u/\lambda \leq t e_0 a_t/\lambda$.

Suppose $x_{i_t}$ is of form M, and let $a_{max} \in
\Gamma^{i_t}_{max}$ be an action of the max player that yields the
greatest fixed point $f^*_{i_t}$ in the max-PPS equation $x_{i_t} =
(P_{*,\tau_{\epsilon}}(x))_{i_t}$.  Then $a_{max}$ belongs to some $B_j$
in Step 4 of the algorithm of Fig. \ref{fig:Qual-LS-Reach}, and thus
there is a $a_{min} \in L_j$ such that the variable
$A_{i_t}(x)_{(a_{max},a_{min})}$ was added earlier to $F$, i.e., it is
variable $x_{i_u}$ for some $u \leq t-1$ or it belongs to $F_0$.  The
probability $p = \tau_{\epsilon}(x_{i_t})(a_{min})$ of this action in
strategy $\tau_{\epsilon}$ is $p=(e_t^2)^{j-1} \cdot (1-e_t^2)/|L_j|$.
All the variables $A_{i_t}(x)_{(a_{max},a)}$ for $a \in \cup_{q=1}^{j}
L_q$ are not in $S$, hence they are all eventually assigned to
$F$. The total probability that strategy $\tau_{\epsilon}$ gives to
the actions $a \in \cup_{q=1}^{j} L_q$ is $1-(e_t^2)^j$, hence the
remaining probability assigned to the other actions $a \in
\Gamma^{i_t}_{min} - \cup_{q=1}^{j} L_q$ is $(e_t^2)^j$ which is $
\leq p e_t$ since $e_t \leq 1/(2N)$.  Therefore, $M_t \leq p M_u +
(1-p - p e_t)M + p e_t$ for some $u \leq t-1$.  By the induction
hypothesis, $M_u \leq a_u M^2 +(1-a_u - g_u)M+ g_u$, where $a_u \geq
\lambda s_u$ and $g_u \leq u e_0 a_u/\lambda$.  Hence, $M_t \leq p
(a_u M^2 +(1-a_u - g_u)M+ g_u) + (1-p - p e_t)M + p e_t$ $= a_t M^2
+(1-a_t - g_t)M+ g_t$, where $a_t = p a_u$ and $g_t = p g_u + p e_t$.
Since $p \geq r_t$ and $a_u \geq \lambda s_u \geq \lambda s_{t-1}$, we
have $a_t \geq \lambda s_t$.  It is easy to check from the definitions
that $e_t \leq e_0 s_{t-1}$.  Indeed, $\log e_t = -d_0 (2N)^t$, while
$\log (e_0 s_{t-1}) = \log e_0 + (2N-1) \sum_{j=1}^{t-1} \log e_j$ $=
-d_0( (2N)^t -2N+1)$.  Since $g_u \leq u e_0 a_u/\lambda$ and $e_t
\leq e_0 s_{t-1} \leq e_0 s_{u} \leq e_0 a_u/\lambda$, we have $g_t =
p g_u + p e_t \leq p (u+1) e_0 a_u /\lambda \leq t e_0 a_t/\lambda$.

Therefore, for all $t$ we have 
$M_t \leq a_t M^2  +(1-a_t - g_t)M+ g_t$,
where $a_t \geq \lambda s_t$ and $g_t \leq t e_0 a_t/\lambda$.
Let $t$ be an index with the maximum $M_t$,
i.e., $M=M_t$.
Then $M \leq a_t M^2  +(1-a_t - g_t)M+ g_t$,
hence $a_t M^2  -(a_t + g_t)M+ g_t \geq 0$.
That is, $(a_t M -g_t)(M-1) \geq 0$.
From Claim \ref{claim:less-1}, $M <1$. Therefore, $a_t M \leq g_t$.
Thus, $M \leq g_t /a_t \leq t e_0/\lambda \leq \epsilon$.
\end{proof}

This concludes the proof of the theorem.

\end{proof}

From the constructions in the proof of the theorem we have the following:

\begin{corollary}
Suppose the algorithm in Figure \ref{fig:Qual-LS-Reach} outputs the set $F$
when it terminates.   Let $S := X-F$.
\begin{enumerate} 
 \item There is a randomized static strategy $\sigma$ for the max player
(maximizing non-reachability) 
  such that for all variables $x_i \in S$, we have $(g_{\sigma, *}^*)_i > 0$. 

\item For all $\epsilon > 0$,
 there is a randomized static strategy
  $\tau_\epsilon$, for the min player (minimizing non-reachability), 
such that for all variables $x_i \in
  F$, $(g^*_{*, \tau_{\epsilon}})_i \leq \epsilon$. 
\end{enumerate}
\end{corollary}
\begin{proof} 
This follows directly from the strategies $\sigma$, and $\tau_\epsilon$,
constructed in the proof of 
Theorem \ref{theorem:Qual-LS-Reach}.
\end{proof}

\bibliographystyle{plainurl}

\end{document}